\documentclass[journal,onecolumn]{IEEEtran}
\usepackage{amsthm}
\usepackage{times,amssymb,amsmath,amsthm,amsfonts,float,nicefrac,color,bbm,mathrsfs,caption,float}
\usepackage{enumerate,multirow,tikz,graphicx,enumitem,soul}
\usepackage[skip=10pt,font=footnotesize]{caption}
\usepackage[mathscr]{eucal}
\usepackage{sidecap}
\usepackage{algpseudocode}
\usepackage{verbatim}
\usepackage{epstopdf}
\usepackage{textcomp}
\usetikzlibrary{shapes,arrows}
\usepackage{mathabx}
\usepackage[noadjust]{cite}
\allowdisplaybreaks
\usepackage[margin=0.05in]{subcaption}
\usepackage[bottom]{footmisc}
\usepackage[normalem]{ulem}

\usepackage[top=1in,bottom=1in,left=1.2in,right=1.2in]{geometry}

\usepackage{hyperref}

\newtheoremstyle{mytheoremstyle}{3pt}{3pt}{\itshape}{}{\bf}{.}{.3em}{} 
\theoremstyle{mytheoremstyle}
\newtheorem{theorem}{Theorem}

\newcommand\nc\newcommand
\nc\bfa{{\boldsymbol a}}\nc\bfA{{\boldsymbol A}}\nc\cA{{\mathscr A}}
\nc\bfb{{\boldsymbol b}}\nc\bfB{{\boldsymbol B}}\nc\cB{{\mathscr B}}
\nc\bfc{{\boldsymbol c}}\nc\bfC{{\boldsymbol C}}\nc\cC{{\mathscr C}}
\nc\bfd{{\boldsymbol d}}\nc\bfD{{\boldsymbol D}}\nc\cD{{\mathscr D}}
\nc\bfe{{\boldsymbol e}}\nc\bfE{{\boldsymbol E}}\nc\cE{{\mathscr E}}
\nc\bff{{\boldsymbol f}}\nc\bfF{{\boldsymbol F}}\nc\cF{{\mathscr F}}
\nc\bfg{{\boldsymbol g}}\nc\bfG{{\boldsymbol G}}\nc\cG{{\mathscr G}}
\nc\bfh{{\boldsymbol h}}\nc\bfH{{\boldsymbol H}}\nc\cH{{\mathscr H}}
\nc\bfi{{\boldsymbol i}}\nc\bfI{{\boldsymbol I}}\nc\cI{{\mathscr I}}
\nc\bfj{{\boldsymbol j}}\nc\bfJ{{\boldsymbol J}}\nc\cJ{{\mathscr J}}
\nc\bfk{{\boldsymbol k}}\nc\bfK{{\boldsymbol K}}\nc\cK{{\mathscr K}}
\nc\bfl{{\boldsymbol l}}\nc\bfL{{\boldsymbol L}}\nc\cL{{\mathscr L}}
\nc\bfm{{\boldsymbol m}}\nc\bfM{{\boldsymbol M}}\nc{\cM}{{\mathscr M}}
\nc\bfn{{\boldsymbol n}}\nc\bfN{{\boldsymbol N}}\nc\cN{{\mathscr N}}
\nc\bfo{{\boldsymbol o}}\nc\bfO{{\boldsymbol O}}\nc\cO{{\mathscr O}}
\nc\bfp{{\boldsymbol p}}\nc\bfP{{\boldsymbol P}}\nc\cP{{\mathscr P}}\nc\eP{{\EuScriptP}}\nc\fP{{\mathfrak P}}
\nc\bfq{{\boldsymbol q}}\nc\bfQ{{\boldsymbol Q}}\nc\cQ{{\mathscr Q}}
\nc\bfr{{\boldsymbol r}}\nc\bfR{{\boldsymbol R}}\nc\cR{{\mathscr R}}
\nc\bfs{{\boldsymbol s}}\nc\bfS{{\boldsymbol S}}\nc\cS{{\mathscr S}}
\nc\bft{{\boldsymbol t}}\nc\bfT{{\boldsymbol T}}\nc\cT{{\mathscr T}}
\nc\bfu{{\boldsymbol u}}\nc\bfU{{\boldsymbol U}}\nc\cU{{\mathscr U}}
\nc\bfv{{\boldsymbol v}}\nc\bfV{{\boldsymbol V}}\nc\cV{{\mathscr V}}
\nc\bfw{{\boldsymbol w}}\nc\bfW{{\boldsymbol W}}\nc\cW{{\mathscr W}}
\nc\bfx{{\boldsymbol x}}\nc\bfX{{\boldsymbol X}}\nc\cX{{\mathscr X}}
\nc\bfy{{\boldsymbol y}}\nc\bfY{{\boldsymbol Y}}\nc\cY{{\mathscr Y}}
\nc\bfz{{\boldsymbol z}}\nc\bfZ{{\boldsymbol Z}}\nc\cZ{{\mathscr Z}}

\newtheorem{lemma}[theorem]{Lemma}
\newtheorem{claim}[theorem]{Claim}

\newtheorem{proposition}[theorem]{Proposition}

\newtheorem{definition}{Definition}

\newtheorem{example}{Example}

\theoremstyle{remark}

\newcommand{\cev}[1]{\reflectbox{\ensuremath{\vec{\reflectbox{\ensuremath{#1}}}}}}

\DeclareMathOperator{\rank}{rank}
\DeclareMathOperator{\supp}{supp}

\newcommand{\ff}{{\mathbb F}}

\begin{document}
		\title{Cyclic and convolutional codes with locality}
		\author{\IEEEauthorblockN{Zitan Chen} \hspace*{1in}
\and \IEEEauthorblockN{Alexander Barg}}
		\date{}
		\maketitle

\begin{abstract}
Locally recoverable (LRC) codes and their variants have been extensively studied in recent years. In this paper we focus on cyclic constructions
of LRC codes and derive conditions on the zeros of the code that support the property of hierarchical locality. As a result, we obtain
a general family of hierarchical LRC codes for a new range of code parameters. We also observe that our approach enables one to represent
an LRC code in quasicyclic form, and use this representation to construct tail-biting convolutional LRC codes with locality.
Among other results, we extend the general approach to cyclic codes with locality to multidimensional cyclic codes, yielding new families of LRC 
codes with availability, and construct a family of $q$-ary cyclic hierarchical LRC codes of unbounded length.
\end{abstract}
{\small{{\tableofcontents}}}

\renewcommand{\thefootnote}{\arabic{footnote}}
\setcounter{footnote}{0}		
		
{\renewcommand{\thefootnote}{}\footnotetext{
\vspace*{-.15in}

\noindent\rule{1.5in}{.4pt}

An extended abstract of this work appears in Proceedings of the 2020 IEEE International Symposium on Information Theory (ISIT 2020).

{Zitan Chen is with Dept. of ECE and ISR, University of Maryland, College Park, MD 20742. Email: chenztan@umd.edu. 

Alexander Barg is with Dept. of ECE and ISR, University of Maryland, College Park, MD 20742 and also with IITP, Russian Academy of Sciences, 127051 Moscow, Russia. Email: abarg@umd.edu. 

This research was supported by NSF grants CCF1618603 and CCF1814487.
}}}
\renewcommand{\thefootnote}{\arabic{footnote}}
\setcounter{footnote}{0}

\section{Introduction}\label{sec:introduction}
\subsection{The locality problem: main definitions}\label{sec:IA}

Locally recoverable (LRC) codes form a family of erasure codes, motivated by applications in distributed storage, that support
repair of a failed storage node by contacting a small number of other nodes in the cluster. LRC codes can be constructed in a number of ways. 
A connection between LRC codes and the well-studied family of Reed-Solomon (RS) codes was put forward in \cite{tb14}, where codes with
large distance were constructed as certain subcodes of RS codes. The results of \cite{tb14} paved the way for using powerful algebraic techniques
of coding theory for constructing other families of LRC codes including algebraic geometric codes \cite{BTV17,Li17a,LiMaXing17b}. 
In \cite{CyclicLRC16} it was observed that a particular class of the codes in \cite{tb14} can be represented in cyclic form, and the 
distance and locality properties of cyclic LRC codes were described in terms of the zeros of the code. This established a framework
for cyclic LRC codes that was advanced in a number of ways in several recent works \cite{chen17,holzbaur2018cyclic,beemer18,luo2018optimal}.

In this paper we focus on several aspects of cyclic LRC codes that have not been previously addressed in the literature. The first of these is 
codes with hierarchical locality (H-LRC codes), defined in \cite{SAC2015}, which assume that the code can correct one or more erasures by using a 
subset whose size depends on the number of erased locations, increasing progressively with their count. In addition to defining the problem and 
deriving a bound on the parameters of H-LRC codes, the authors of \cite{SAC2015} extended the construction of \cite{tb14} to the hierarchical case. 
Their construction was further generalized to algebraic geometric codes in \cite{bbv19}.

Another general problem related to cyclic codes that we study addresses construction of LRC convolutional codes. This class of erasure codes was 
considered in several previous works \cite{IKZ2018,ZLLS2016,Datta2014} before being thoroughly analyzed in a recent paper \cite{MN19}. Here we 
exploit a classic link between quasicyclic codes and convolutional codes \cite{SvT79} to construct LRC convolutional codes with a particular type 
of locality whose parameters are controlled by the set of zeros of the underlying (quasi)cyclic code.

We also address the problem of maximum length of optimal LRC codes \cite{GXY19}, wherein the main question is constructing such codes of 
length larger than the size of the code alphabet. This problem has been the subject of a number of recent papers including 
\cite{B+17,beemer18,LiMaXing17b}, and we study this question in the hierarchical case. Our techniques here combine
the approach to cyclic H-LRC codes developed in this paper with the ideas of \cite{luo2018optimal} aimed at constructing long LRC codes.

We describe the known and the new results in more details after giving the basic definitions. In coding-theoretic terms, the central problem addressed by LRC codes is correcting one or several erasures in 
the codeword based on the contents of a 
small number of other coordinates. This problem was first isolated in \cite{gopalan2012locality}, and it has been actively studied in the last 
decade. 

\begin{definition}[\sc LRC codes]
	\label{def:lrc}
A linear code $\cC\subset \mathbb{F}_q^n$ is locally recoverable with locality $r$ if for every $i\in \{1,2,\dots,n\}$
there exists an $r$-element subset $I_i\subset \{1,2,\dots,n\}\backslash \{i\}$
	and a linear function $\phi_i:\mathbb{F}_q^r\to\mathbb{F}_q$ such that for every codeword $c\in\mathcal{C}$ we have
	$
	c_i=\phi_i(c_{j_1},\dots,c_{j_r}),
	$
	where $j_1<j_2<\cdots   <j_r$ are the elements of $I_i.$ 
\end{definition}
The coordinates in $I_i$ are called the \emph{recovering set} of $i$, and the set $\{i\}\cup I_i$ is called a \emph{repair group}. 
Below we refer to a linear LRC code of length $n$, dimension $k$, and locality $r$ as an $(n,k,r)$ LRC code. Since the code is
occasionally used to correct a large number of erasures (such as in the event of massive system failure), another parameter of interest is
the maximum number of erasures that it can tolerate. This is controlled by the minimum distance $d(\cC)$ of the code, for which there are
several bounds known in the literature. We will be interested in the generalized Singleton bound of \cite{gopalan2012locality} which
states that for any LRC code $\cC$,	
    \begin{equation}\label{eq:sb-1}
	d(\cC)\leq n-k-\left\lceil \frac{k}{r}\right\rceil+2.
    \end{equation}

While in most situations repairing a single failed node restores the system to the functional state, occasionally there may be a need to 
recover the data from several concurrent node failures. The following extension of the previous definition is due to \cite{kamath2012codes}. 
\begin{definition}[\sc $(r,\delta)$ locality] \label{def:r-delta}
For any $\delta\ge 2$ we say that a linear code $\cC$ has {\em $(r,\delta)$ locality} if every coordinate $i\in\{1,\dots,n\}$ is contained in a subset $J_i\subset \{1,\ldots,n\}$ of size at most $r+\delta -1$ 
such that the restriction $\cC_{J_i}$ to the coordinates in $J_i$ forms a code of distance at least $\delta$.
\end{definition}
Note that in the case of $\delta=2$ the codes defined here are exactly the codes of Def.~\ref{def:lrc} above.

 An intermediate situation arises when the code is designed to correct a single erasure by 
contacting a small number $r_1$ of helper nodes, while at the same time supporting local recovery of multiple erasures. Extending this
idea to multiple levels of local protection, the authors of \cite{SAC2015} introduced the concept of hierarchical LRC (H-LRC) codes, which are 
defined as follows.

\begin{definition}[\sc H-LRC codes]\label{def:H-LRC}
	Let $h\geq 1$, $0 < r_1 < r_2 < \ldots < r_h < k$, and $1 < \delta_1 \leq \ldots \leq \delta_h \leq d$ be integers. 
	A linear code $\cC\subset \mathbb{F}_q^n$ is said to have $h$-level hierarchical locality $(r_1,\delta_1),\ldots,(r_h,\delta_h)$ if for every $1\leq i\leq h$ and every coordinate of the code $\cC$ there is a punctured code $\cC^{(i)}$ such that the coordinate is in the support of $\cC^{(i)}$ and\\ 
	\indent (a) $\dim(\cC^{(i)})\le r_i$,\\
	\indent (b) $d(\cC^{(i)})\ge\delta_i$,\\
	\indent (c) the $i$-th local code $\cC^{(i)}$ has $(i-1)$-level hierarchical locality $(r_1,\delta_1),\ldots,(r_{i-1},\delta_{i-1})$.
\end{definition}

The authors of \cite{SAC2015} proved the following extension of the bound \eqref{eq:sb-1}: The minimum distance of an $h$-level
H-LRC code with locality satisfies the inequality
  \begin{align} \label{eq:sb}
d \leq n-k+\delta_h - \sum_{i=1}^{h}\left\lceil\frac{k}{r_i}\right\rceil(\delta_i-\delta_{i-1}),
  \end{align} 
  where $\delta_0=1$. In particular, for $h=1$ this gives a version of the bound \eqref{eq:sb-1} for the distance of an $(n,k)$ code
with $(r,\delta)$ locality:
 \begin{equation}\label{eq:delta}
   d\le n-k+\delta-\Big\lceil\frac kr\Big\rceil(\delta-1).
 \end{equation}
We call an H-LRC code {\em optimal} if its distance attains the bound \eqref{eq:sb}. Note that it is possible that the code is optimal while its
local codes $\cC^{(i)}$ for some or even all $i\le h-1$ are not. We say that an H-LRC code $\cC$ is {\em strongly optimal} if 
in every level $i,1\leq i\leq h$, the $i$-th local codes are optimal H-LRC codes with $i-1$ levels. 
For $h=1$ the distance of the optimal code with $(r,\delta)$ locality attains the bound \eqref{eq:delta} with equality. 
  
\subsection{Earlier results and our contributions}

As noted above, the starting point of our constructions is a cyclic version of the RS-like codes with locality designed in \cite{tb14}. 
RS codes over $\ff_q$ can be alternatively described in terms of polynomial evaluation and (in the case that the code length $n$ divides
$q-1$) as cyclic codes of the BCH type. While \cite{tb14} adopted the former approach, the authors of \cite{CyclicLRC16} studied the
cyclic case, finding a condition on the zeros of the code that support the locality property. 

The approach of \cite{CyclicLRC16} was later advanced in a number of ways. In particular, the authors of \cite{chen17} extended the main 
construction of \cite{CyclicLRC16} to codes with $(r,\delta)$ locality, designing a cyclic representation of the polynomial evaluation codes
from \cite{tb14} for the general case of $\delta\ge 2.$

Our main results relate to constructions of cyclic LRC codes with hierarchy. Several recent studies presented families of codes with 
multiple levels of erasure correction \cite{freij2016locally,grezet2019complete,yang2019hierarchical}, not necessarily within the framework of the 
above definition. As far as H-LRC codes are concerned, the only general family of optimal H-LRC codes that we are aware of was presented in \cite{SAC2015,bbv19}. This construction essentially followed the approach of \cite{tb14}, relying on 
multivariate polynomials that are constant on the blocks that form the support of the code $\cC^{(i)}$ in Def.~\ref{def:H-LRC}.
The codes of \cite{SAC2015,bbv19} form a family of strongly optimal H-LRC codes which can be 
constructed for any code length $n\le q,$ dimension $k,$ and any values of $r_i,i=1,\dots,h$ as long as $r_i|r_{i+1}, i=1,\dots,h-1$ and $r_h|k.$
The divisibility constraint is essential for the constructions discussed, and it limits the possible choices of the code parameters.

In this paper we derive conditions on the zeros of a cyclic code that support several levels of hierarchy. As a result, we construct 
families of cyclic H-LRC codes that do not rely on the divisibility assumptions, which yields a new range of code parameters. We also derive 
conditions that are sufficient for our codes to be strongly optimal and give examples of such codes. These results are presented in Sec.~\ref{sec:hlrc}.

These results also enable us to connect the construction of LRC cyclic codes and convolutional codes with locality. 
The recent work \cite{MN19} focused on the so-called sliding window repair property of convolutional codes. The authors of \cite{MN19}
further observed that certain families of convolutional codes, notably the so-called codes with the maximum distance profile 
\cite{gluesing2006strongly,TRV2012}, suggest an approach to constructing codes with locality. They also 
presented a family of LRC convolutional codes with sliding window repair for the case of {\em column locality}, and they suggested
that there may be a connection between H-LRC codes and LRC convolutional codes. We show that this connection indeed leads to
fruitful results, designing LRC convolutional codes for the case of {\em row locality} defined herein. The lower bounds on the column distance of 
the codes constructed here and in \cite{MN19} are the same; however the alphabet size of our codes is much smaller than in \cite{MN19}.
We also derive an upper bound on the column distance of LRC convolutional codes (with either type of locality); however, our construction falls short of attaining it. The method that we use relies on the characterization of zeros of cyclic block H-LRC codes described above. We observe that several levels of hierarchy enable one to put our cyclic LRC codes in quasicyclic form, and then use a classic connection between quasicyclic codes and convolutional codes \cite{SvT79} to construct convolutional codes with locality; see Sec.~\ref{sec:cc}.

We also examine two other problems for LRC codes that benefit from the cyclic code construction. The first of them is 
the problem of maximum length of optimal LRC codes put forward in \cite{GXY19}. 
Answering the challenge of constructing optimal LRC codes of length larger than $q,$
the authors of \cite{CaiMiaoSchwartzTang18,GXY19,beemer18,luo2018optimal} constructed several families of optimal cyclic codes of large, and in some cases even unbounded length, and \cite{CaiMiaoSchwartzTang18} extended these results to the case of several erasures. Here we follow the lead of 
\cite{luo2018optimal} and construct an infinite family of H-LRC codes over a given finite field and establish conditions for their optimality in 
terms of the bound \eqref{eq:sb}. Finally, we consider the problem of {\em availability} which calls for constructing LRC codes with several disjoint recovering sets for each code coordinate (see the definition in Sec.~\ref{sec:avl}). We note that multidimensional cyclic codes naturally 
yield several recovering sets for the coordinates. We use a description of bi-cyclic codes in terms of their zeros 
together with a special version of code concatenation \cite{saints1993hyperbolic} to construct codes with availability and rate higher than the rate 
of product codes. We note that the known bounds on the code parameters for multiple recovering sets 
\cite{rawat2016locality,tamo2016bounds,wang2014repair} do not support a conclusive picture, and we are not aware of general 
families of codes with availability whose distance attains one of the known upper limits. Our construction in this paper also falls in the same 
category.

Although we do not pursue this direction here, let us note that the methods of this paper enable one to construct codes with both hierarchical 
locality and availability. We remark that constructions of LRC codes that have both properties were presented in \cite{bbv19}, where
the main tools were fiber products and covering maps of algebraic curves.

\subsection{Optimal cyclic LRC codes}\label{sec:clrc}
Let $\cC$ be a cyclic code of length $n$ with generator polynomial $g(x)$ and check polynomial $h(x)=\frac{x^n-1}{g(x)}.$
The dual code $\cC'$ has generator polynomial $g_{\cC'}(x)=x^{\deg(h(x))}h(x^{-1}),$ and the code $\cev{\cC'}$ 
generated by $h(x)$ is obtained from $\cC'$ by inverting the order of the coordinates. A codeword $a(x)\in\cC'$ of weight $r+1$ defines 
a repair group of the code $\cC,$ and so does the reversed codeword $\cev{a}(x)\in\cev{\cC'}.$ 
For this reason below in this section we argue about the code $\cev{\cC'}$ rather than $\cC'$, which makes the writing more compact without affecting the conclusions.

Let us recall a connection between cyclic codes and LRC codes of \cite{tamo2016cyclic}, which we present in the form close 
to the earlier works \cite{beemer18,luo2018optimal}. The following lemma underlies constructions of cyclic LRC codes in this paper and
elsewhere, and it represents a mild extension of Lemma 3.3 in \cite{tamo2016cyclic}. In the statement as well as elsewhere in the paper we do not distinguish between zeros of the code and their exponents
in terms of some fixed primitive $n$th root of unity in $\ff_q$.
\begin{lemma}\label{lemma:zeros} Let $\cC$ be a cyclic code over $\ff_q$ of length $n|(q-1)$ and let $\alpha$ be a primitive $n$th root of unity in $\ff_q.$ Suppose that $n=\nu m$ for some integers $\nu,m.$ Then the code $\cev{\cC'}$ contains a vector
   \begin{equation}\label{eq:b}
   b(x)=\sum_{i=0}^{m-1}x^{i\nu}\alpha^{(m-1-i)\nu u}, \quad u\in\{0,\dots,m-1\}
   \end{equation}
if and only if the set $\cL=\{{u+im},i=0,1,\dots,\nu-1\}$ is among the zeros of $\cC.$ 
\end{lemma}

\begin{proof}
Notice that the polynomial $b(x)$ can be equivalently written as
   $
   b(x)=\frac{x^n-1}{x^\nu-\alpha^{\nu u}},
   $
where 
   \begin{equation}\label{eq:ap}
x^\nu-\alpha^{\nu u}=\prod_{i=0}^{\nu-1}(x-\alpha^{u+im})
  \end{equation}
is the annihilator polynomial of the set $\cL.$ Thus, $b(\alpha^t)\ne 0$ for all $t\in \cL.$ If $b(x)\in \cev{\cC'},$ this implies that $\cL$ is a subset of the set of zeros of $\cC.$

Conversely, let $g (x)=(x^\nu-\alpha^{\nu u})p(x)$ be the generator polynomial of $\cC$.
 Then
   $$
   \Big(\frac{x^n-1}{g (x)}\Big)p(x)=\frac{x^n-1}{x^\nu-\alpha^{\nu u}}=b(x)\in \cev{\cC'}.
   $$
\end{proof}

This lemma immediately yields the cyclic codes from \cite{tamo2016cyclic} (the cyclic case of the codes from \cite{tb14}).
\begin{theorem}[\cite{tamo2016cyclic}] \label{thm:cyclic} Let $(r+1)|n, r|k, n|(q-1)$. Let $\alpha\in \ff_q$ be a primitive $n$-th root of unity, and let $\cC$ be an $(n,k)$ cyclic code with zeros
$\alpha^i, i\in \cZ:=\cL \cup\cD,$ where
  \begin{equation}\label{eq:zeros}
    \begin{aligned}
        \cL&=\{1+l(r+1),l=0,\dots,\textstyle{\frac n{r+1}}-1\}\\
        \cD&=\{1,2,\dots,n-\frac kr(r+1)+1\}.
        \end{aligned}
  \end{equation}
  Then $\cC$ is an $(n,k,r)$ optimal LRC code.
\end{theorem}
\begin{proof}
Since $\cL\subset \cZ,$ Lemma \ref{lemma:zeros} implies that 
  $$
  b(x)=\sum_{i=0}^r \alpha^{i\frac{n}{r+1}}x^{(r-i)\frac{n}{r+1}}
  $$
is a codeword in $\cev{\cC'}.$ This codeword is of weight $r+1,$ and its cyclic shifts give $\frac n{r+1}$ disjoint repair groups, supporting the
locality claim. At the same time, the BCH bound implies that $d(\cC)\ge n-\frac kr(r+1)+2,$ so the code is optimal by \eqref{eq:sb-1} once one observes $|\cZ|=n-k$ and $\dim(\cC)=n-|\cZ|=k$.
\end{proof}

This construction extends without difficulty to codes with $(r,\delta)$ locality for any $\delta\ge 2.$ A family of optimal codes in the sense of the bound \eqref{eq:delta} was constructed in \cite{tb14}, Construction 8 (see also \cite{chen17}). The codes in this family are constructed as certain
subcodes of Reed-Solomon codes that rely on piecewise-constant polynomials over $\ff_q.$ In the particular case that the code length
$n$ divides $q-1$ it is possible to represent these codes in cyclic form. For this, we assume that $r|k,$ take $m=r+\delta-1,$ and take the zeros of the code to be 
  \begin{equation}\label{eq:rts}
  \begin{aligned}
  \cL&=\{i+lm\mid l=0,\dots,\nu-1, i=1,\dots,\delta-1\}\\ &\cD=\{1,2,\dots,n-k-((k/r)-1)(\delta-1)\}.
  \end{aligned}
  \end{equation}
As will be apparent from the proof of Lemma \ref{le:local0}, the condition about zeros given by the set $\cL$ translates into conditions on the dual code that support the locality claim. The distance of the code $\cC$ clearly meets the bound \eqref{eq:delta} with equality. 

\subsection{Cyclic codes with locality}\label{sec:construction}
In the next lemma we present a slightly more general view of the method in Theorem \ref{thm:cyclic} that will be instrumental in the code constructions below in this work. The main element of the construction is code $\cC^{(0)}$ defined in \eqref{eq:C0}, which isolates a repair group in the code $\cC$ and supports local correction of several erasures.
\begin{lemma} 
	\label{le:local0}
Let $n|(q-1), n=\nu m$. Let $\alpha$ be a primitive $n$-th root of unity in $\ff_q$, and fix $\delta\in\{2,\dots,m\}.$ Let $\cZ$ be a subset of size $Z$ such that
  $$
	\{1,\ldots,\delta-1\}\subset \cZ \subset \{0,\ldots,{m}-1\}
  $$
and let $\cL = \bigcup_{s=0}^{\nu-1}(\cZ+ s{m}).$
Consider a cyclic code $\cC=\langle g(x)\rangle$ of length $n,$ where 
   $$	
	g(x)=p(x)\prod_{t\in\cL}(x-\alpha^t) , 
   $$
and $p(x)\in \ff_q[x]$ is some polynomial. Let 
  \begin{equation}\label{eq:C0}
   \cC^{(0)}=\{ (c_0, c_{\nu},\ldots,c_{({m}-1)\nu}) \mid (c_0,\ldots,c_{n-1})\in \cC \}.
  \end{equation}
Then  $
	\dim(\cC^{(0)})\leq {m}-Z, \quad d(\cC^{(0)})\geq \delta,
  $
and thus, the code $\cC$ is an LRC code with $({m}-Z,\delta)$ locality.

Further, if for every $u\in\{0,\ldots,m-1\}\setminus\cZ$ 
there exists $s\in\{0,1,\dots,\nu-1\}$ such that $g(\alpha^{u+sm})\neq 0$, then $\dim(\cC^{(0)})=m-Z$.
\end{lemma}
\begin{proof} We proceed similarly to Theorem \ref{thm:cyclic}. Let 
   $$
   l(x)=\prod_{t\in \cL}(x-\alpha^t)=\prod_{s=0}^{\nu-1}\prod_{i\in\cZ}(x-\alpha^{i+sm})=\prod_{i\in\cZ} l_i(x)
   $$
   where $l_i(x):=x^\nu-\alpha^{\nu i}.$ Let $h(x)=\frac{x^n-1}{g(x)}$ and consider a subset of $Z$ codewords of $\cev{\cC}'$ given by
   $$
   b_i(x):={h}(x)p(x)\prod_{j\in \cZ\backslash\{ i\}}l_j(x), \quad i\in\cZ.
   $$
A codeword $b_i$ has the form 
   $$
   b_i(x)=\frac{x^n-1}{l_i(x)}=\sum_{j=0}^{m-1} \alpha^{(m-1-j)i\nu} x^{j\nu},
   $$
Hamming weight $m$, and contains $\nu-1$ zero coordinates after every nonzero entry. 

To prove the statement about locality, let us form a $Z\times m$ matrix $H$ obtained by 
inverting the order of coordinates in the codewords $b_i, i\in \cZ,$ writing the resulting vectors as rows, and discarding all the zero columns. 
By construction, every row of $H$ is a parity-check equation of the code $\cC^{(0)}.$ Any submatrix of $\delta-1$ columns of $H$ has rank $\delta-1$ (its first $\delta-1$ rows form a Vandermonde determinant),
and thus, $d(\cC^{(0)})\ge \delta.$ Since the rows of $H$ give $Z$ independent parity-check equations for the code $\cC^{(0)},$ we also have
$\dim(\cC^{(0)})\le m-Z.$ 
	This argument exhibits a local code in the coordinates that are integer multiples of $\nu,$ and by cyclic shifts we can partition the 
	set $\{0,1,\dots,n-1\}$ into supports of disjoint local codes of length ${m}$ and distance at least $\delta.$ Furthermore, we note that the punctured code $\cC^{(0)}$ is itself a cyclic code of length $m|n.$ Let $g_{0}(x)$ be its generator polynomial. Since 
each row of $H$ is a parity-check equation for the code $\cC^{(0)}$, we have $g_{0}(\alpha^{i\nu})=0$ for every $i\in\cZ$ and thus 
$\deg(g_0(x))\geq  Z $. Together these arguments prove the claim 	about $(m-Z,\delta)$ locality of the code $\cC.$

Next we show that if $\deg(g_0(x))>  Z ,$ then necessarily there exists $u\in\{0,\ldots,m-1\}\setminus\cZ$ such that $g(\alpha^{u+sm})=0$ for all $s=0,\dots,\nu-1.$ 
Suppose that there exists $u\in\{0,\ldots,m-1\}\setminus\cZ$ such that $g_{0}(\alpha^{u\nu})=0.$ 
By Eq.~\eqref{eq:ap} in the proof Lemma~\ref{lemma:zeros} there exists a codeword $b_u\in\cev{\cC}'$ given by 
\begin{align*}
	b_u(x)=\sum_{j=0}^{m-1} \alpha^{ju\nu} x^{(m-1-j)\nu}=\frac{x^n-1}{x^\nu-\alpha^{\nu u}}.
\end{align*} 
On the other hand, $\cev{\cC}'=\langle h(x)\rangle,$ so $h(x)|b_u(x)$ and therefore $(x^\nu-\alpha^{\nu u})|g(x)$. 
Noticing that $x^\nu-\alpha^{\nu u}=\prod_{s=0}^{\nu-1}(x-\alpha^{u+sm})$ (cf.~\eqref{eq:ap}), we conclude that $g(x)$ is divisible by $x-\alpha^{u+sm}$ for all $s=0,\ldots,\nu-1$. Hence, if for every $u\in\{0,\ldots,m-1\}\setminus\cZ$ there exists $0\leq s\leq \nu-1$ such that $(x-\alpha^{u+sm})\nmid g(x)$, i.e., $g(\alpha^{u+sm})\neq 0$, then $\deg(g_0(x))= Z $, and thus $\dim(\cC^{(0)})=m- Z $.
\end{proof}

\section{Codes with hierarchical locality}\label{sec:h}

	\subsection{Optimal cyclic codes with hierarchy}\label{sec:hlrc}
In this section we construct a family of H-LRC cyclic codes with $h\geq 1$ levels of hierarchy and derive sufficient conditions for their
optimality. Suppose that $h$ is fixed and we are given the local dimension $r_1$ (the dimension of the first, innermost local code), and
the designed local distances $1=\delta_0\le \delta_1 \leq \ldots \leq \delta_h \leq \delta_{h+1}$, 
where $\delta_{h+1}$ is the designed distance of the overall code $\cC.$

We will assume that the first local code is MDS and thus its length is $n_1=r_1+\delta_1-1$. For $1\leq i\leq h$, let $n_{i+1}=\nu_i n_i$ be the length of the code in the $(i+1)$st level of hierarchy, where $\nu_i>1$ is an integer.
Let $\mathbb{F}_q$ be a finite field and suppose that $n_{h+1}|(q-1)$. 

We construct a cyclic H-LRC code $\cC$ over $\mathbb{F}_q$ of length $n=n_{h+1}$ and designed (local) distances $\delta_1,\ldots,\delta_{h+1}$ as follows.
Let $\alpha\in \mathbb{F}_q$ be a primitive $n$-th root of unity. 
The code $\cC$ will be given by its defining set of zeros $\cZ$ which we specify via a recursive procedure. 
Consider the set of 
exponents $\cD_1=\{1,\ldots,\delta_1-1\}$ of the primitive element $\alpha$. Further, let $\cL_1=\emptyset$ and
\begin{align}
\cZ_1=\cL_1\cup\cD_1.\label{eq:def-set-1}
\end{align}
Having \eqref{eq:zeros} in mind, for $1\leq i\leq h$ let \vspace*{-.1in}
     \begin{equation} \label{eq:def-set-2}
     \begin{aligned}
\cL_{i+1} = \bigcup_{s=0}^{\nu_i-1}(\cZ_i+ &sn_i), \quad
\cD_{i+1} = \{1,\ldots, \delta_{i+1}-1\},\\[-.05in]
&\cZ_{i+1} = \cL_{i+1}\cup\cD_{i+1}.
\end{aligned}
     \end{equation}
Finally, put $\cZ=\cZ_{h+1}$. 

The generator polynomial of the cyclic code $\cC$ of length $n$ is given by
\begin{align}
g(x) = \prod_{t\in\cZ}(x-\alpha^t).\label{eq:hlrc-def}
\end{align}

The parameters of the code $\cC$ are estimated as follows.

\begin{proposition}
	\label{prop:local}
	(i) The dimension of the code $\cC$ is $n-|\cZ|$ and the distance $d\ge \delta_{h+1}.$
	(ii) 
	The code $\cC$ is an $h$-level H-LRC code with locality
	$(n_i-|\cZ_i|,\delta_{i}),i=1,\dots,h.$
\end{proposition}
\begin{proof}
	$(i)$ The value of the dimension is clear from the construction, and the estimate for the distance comes from the BCH bound.
	
	$(ii)$ The statement about the locality parameters follows by Lemma~\ref{le:local0} 
once we observe that $g(x)$ is divisible by $\prod_{t\in\bigcup_{s=0}^{n/n_i-1}(\cZ_i+sn_i)}(x-\alpha^t)$ for every $i=1,\dots,h$.
\end{proof}

Next we examine the conditions that suffice for the distance of $\cC$ to meet the bound \eqref{eq:sb} with equality. 
We build up the optimality of our code in an incremental manner in the sense that we first ensure that the first local codes are optimal 
(i.e., MDS codes), and relying on these optimal local codes we make sure that the second local codes are optimal 
(i.e., optimal LRC codes), and so forth until we reach the outermost code. 

Let $r_1 < r_2 < \ldots < r_h<r_{h+1}=\dim(\cC)$ be the chosen values of the dimensions of the local codes.
 As before, we set $n_1=r_1+\delta_1-\delta_0$ 
and let $n_{i+1}=\nu_i n_i$ for $1\leq i\leq h$ where the integer number $\nu_i$ satisfies $\nu_i\geq \lceil r_{i+1}/r_i \rceil$.
Note that this assumption does not entail a loss of generality since, assuming the $i$th and $(i+1)$st local codes are optimal, by \eqref{eq:def-set-2} we have
 $|\cL_{i+1}|=(n_i-r_i)\nu_i$, and $r_{i+1}=n_{i+1}-|\cZ_{i+1}|\le n_{i+1}-|\cL_{i+1}|=r_i\nu_i.$

To show optimality, we connect the target values of the local distances $\delta_2,\delta_3,\ldots,\delta_{h+1}$ with the dimension values
through several auxiliary parameters. For $1\leq i\leq h$, let us write $r_{i+1} = a_i r_i -b_i$ where $0\leq b_i< r_i$.
Further, let $b^{(i)}_i=b_i$ and for $j=i,i-1,\ldots,2$, let 
\begin{align*}
 b_j^{(i)}=u_{j-1}^{(i)} r_{j-1}+b_{j-1}^{(i)}, \quad 0\le b_{j-1}^{(i)}<r_{j-1}.
\end{align*}
Put $b^{(0)}_0=0$ and for $1\leq i\leq h$ let
\begin{align}
b^{(i)}_0 = (b^{(i)}_1 + b^{(i-1)}_0) \bmod r_1, \quad u^{(i)}_0 = \left\lfloor \frac{b^{(i)}_1 + b^{(i-1)}_0 }{r_1}\right\rfloor.\label{eq:b0}
\end{align}
Finally, for $1\leq i\leq h$, let \vspace*{-.1in}
\begin{align}
\delta_{i+1} = (\nu_i-a_i)n_i+\delta_i+\sum_{j=1}^{i-1}u^{(i)}_j n_j + u^{(i)}_0 n_1+b^{(i)}_0-b^{(i-1)}_0.\label{eq:d}
\end{align}
The high-level ideas behind these parameters are as follows. 
By Lemma~\ref{le:local0}, the quantities $\delta_1,\ldots,\delta_{h+1}$ control the distances via the BCH bound and we would like to make these 
quantities as large as possible given the target dimensions. We need to make sure that the $i$th local code has a run of consecutive zeros  
of length $\delta_i-1,$ and our budget of creating such a run is limited by the dimension. 
Therefore, we seek to rely on the already present runs of zeros of the $j$th local codes, $j<i$, and spend the budget 
frugally on the way to optimality. 
The quantities $b^{(i)}_j$ serve the purpose bridging the ``distance gap'' (the gaps between runs of zeros) between the local codes in
levels $j$ and $j+1$ on the way to ensure the distance of the $i$th code.  

As for the dimensions of the local codes, by Lemma~\ref{le:local0} they  are determined by $\cZ_i,1\leq i\leq h+1$. 
The cardinality of $\cZ_i$ is established in the next claim.
\begin{proposition}
	\label{cl:card}
	For $1\leq i\leq h+1$, we have $|\cZ_i|=n_i-r_i$.
\end{proposition}
The proof of this proposition proceeds by induction and is given in Appendix~\ref{app:card}. An examination of the proof
also gives a better understanding of the parameters $a_i,b_i$ introduced above.

On account of Proposition~\ref{cl:card}, the locality parameters of the code $\cC$ are $(r_i,\delta_i),1\leq i\leq h.$ 
Furthermore, $\dim{(\cC)}=r_{h+1}$, and by the BCH bound $d(\cC)\geq \delta_{h+1}$.

Sufficient conditions for optimality of the code $\cC$ are given in the following lemma whose proof is given in Appendix~\ref{app:le:opt}.

\begin{lemma}
	\label{le:opt}
Suppose that for $i\geq 2$ and $2\leq s \leq i$ the following conditions are satisfied:
	\begin{equation} \label{eq:opt}
	\begin{aligned}
	\left\lceil\frac{r_{s+1}}{r_s}\right\rceil \left\lceil\frac{r_s}{r_1}\right\rceil-\left\lceil\frac{r_{s+1}}{r_1}\right\rceil &=
	u^{(s)}_0 + u^{(s)}_1 + \sum_{j=2}^{s-1}u^{(s)}_j\left\lceil \frac{r_j}{r_1} \right\rceil\\
	\left\lceil\frac{r_{s+1}}{r_s}\right\rceil \left\lceil\frac{r_s}{r_l}\right\rceil-\left\lceil\frac{r_{s+1}}{r_l}\right\rceil &= u^{(s)}_l + \sum_{j=l+1}^{s-1}u^{(s)}_j\left\lceil \frac{r_j}{r_l} \right\rceil,
	\end{aligned}
	\end{equation}
where the first condition holds for $s\geq 2$ and the second for $2\le l\le s-1.$	
	Then for $1\leq i\leq h$, we have
	\begin{align*}
	\delta_{i+1} = n_{i+1}-r_{i+1}+\delta_{i}-\sum_{l=1}^{i}\left\lceil \frac{r_{i+1}}{r_l} \right\rceil (\delta_l-\delta_{l-1}).
	\end{align*}
\end{lemma}

It follows from Lemma~\ref{le:local0}, Lemma~\ref{le:opt}, and the bound \eqref{eq:sb} that the code $\cC$ is 
an $(n=n_{h+1},k=r_{h+1})$ optimal H-LRC code with local parameters $(r_i,\delta_i),1\leq i\leq h$. 
Clearly, when $h=0$ our construction gives an $(n_1,r_1)$ MDS code and when $h=1,$ it gives an $(n_2,r_2)$ optimal LRC code of \cite{tamo2016cyclic}.
For $h=2$, conditions \eqref{eq:opt} take a simpler form: 
\begin{align}
\left\lceil\frac{r_{3}}{r_2}\right\rceil \left\lceil\frac{r_2}{r_1}\right\rceil-\left\lceil\frac{r_{3}}{r_1}\right\rceil
&= \left\lfloor\frac{b_1+b_2}{r_1} \right\rfloor. \label{eq:h-2}
\end{align}
We note that the condition of \cite[Theorem 2.6]{SAC2015} is easily seen to be equivalent to \eqref{eq:h-2}. Another 
known case of optimality, the divisibility conditions $r_i| r_{i+1}, i=1,\dots,h,$ is also covered by Lemma~\ref{le:opt} 
(in this case both the left-hand sides and the right-hand sides of \eqref{eq:opt} are zero).

Let us give a general example of the choice of parameters that ensures optimality. 
Suppose that $r_1 \geq 2^h$ and $r_{i+1} = 2r_i - 1$ for $1\leq i\leq h$. Then 
conditions \eqref{eq:opt} are satisfied. Indeed, we have $r_i = 2^{i-1}(r_1-1)+1$ for $1\leq i\leq h+1$. Therefore, for $1\leq j\leq i\leq h+1$ we have
     $$
\left\lceil \frac{r_i}{r_j} \right\rceil = \left\lceil 2^{i-j}-\frac{2^{i-j}-1}{2^j(r_1-1)+1} \right\rceil= 2^{i-j},
     $$
where the last equality follows because $r_1 \geq 2^h$. It follows that the left-hand sides of conditions \eqref{eq:opt} are zero. On the other hand, we have $b_i=1$ for all $1\leq i\leq h$. Since $r_1\geq 2^h$, we have $u^{(s)}_l=0$ for all $1\leq s\leq h$ and $0\leq l\leq s-1,$
and thus, the right-hand sides are also zero, which confirms the optimality claim.

\begin{proposition} Suppose that the conditions \eqref{eq:opt} are satisfied, then the code $\cC$ is a strongly optimal H-LRC code in the sense of Sec.~\ref{sec:IA}. 
\end{proposition}
\begin{proof} It suffices to show that the dimension of the $i$-th local code \emph{equals} $r_i$ for all $i$. 
	
	By assumption, we have $r_{i+1}>r_i$ and thus $a_i\geq 2$ for $1\leq i\leq h$. It is not difficult to verify from \eqref{eq:d} that $\delta_{i+1}\leq (\nu_i-a_i+1)n_i \leq n_{i+1}-n_i$ for all $1\leq i\leq h$. 
	We claim that $g(\alpha^{n-n_i+t_i})\neq 0$ for every $t_i\in \cT_i, 1\leq i\leq h$ where
	\begin{align*}
		\cT_i=\{n\} \cup (\{1,\ldots,n_i-1\}\setminus\cZ_i+n-n_i).
	\end{align*}
	Then by the second part of Lemma~\ref{le:local0} the dimension of the $i$-th local code equals $r_i$ and the strong optimality follows.
	
Now let us show $n-n_i+t_i\notin\cZ$. Observe that the set $\cZ$ contains $n/n_i$ copies of $\cZ_i$ and the set $\cT_i$ is the complement to the last 
one of those copies with respect to $\{1,\ldots,n_i\}$. Indeed, we have $n-n_i+t_i=t_i+(n/n_i-1)n_i \notin \{n\}\cup(\cZ_i+(n/n_i-1)n_i)$. Now 
consider the last copy of $\cZ_{i+1}$ contained in $\cZ$. Obviously, it contains the last copy of $\cZ_{i}$. To establish $n-n_i+t_i\notin\cZ$, it 
remains to show $n-n_i+t_i$ is not in the last copy of $\cD_{i+1}$, namely, $n-n_i+t_i\notin \cD_{i+1}+n-n_{i+1}$. Since $\delta_{i+1}\leq n_{i+1}-
n_i$ as we observed above and $n_{i+1}-n_i+t_i\geq n_{i+1}-n_i+\delta_i$, we have $n_{i+1}-n_i+t_i\notin \cD_{i+1}$. 
	It follows that $n-n_i+t_i\notin \cD_{i+1}+n-n_{i+1}$. Therefore, $n-n_i+t_i\notin\cZ$ and $g(\alpha^{n-n_i+t_i})\neq 0$ for every $t_i\in \cT_i, 1\leq i\leq h$. 
	
	In the case that $r_{i+1}=r_i$ for some $1\leq i\leq h$ (although we rule out this trivial case in Definition~\ref{def:H-LRC}), the code $\cC$ is 
still strongly optimal if the optimality conditions are satisfied. In fact, if $r_{i+1}=r_i$ then from \eqref{eq:d} we have $\delta_{i+1}=(\nu_i-
a_i)n_i+\delta_i=n_{i+1}-n_i+\delta_i\notin\cD_{i+1}$. By similar arguments as above, we have $n-n_i+t_i\notin \cZ$ for every $t_i\in \cT_i, 1\leq i
\leq h$, and thus establish the strong optimality of the code.
\end{proof}

We conclude with a numerical example that shows that the assumptions on the parameters can be simultaneously satisfied for moderate
values of the length and alphabet size.
\begin{example}
	Consider the case $h=3$. Let $r_1=2$ and $\delta_1=2$. Then $n_1=3$. Let $(n_2,r_2)=(9,3)$, $(n_3,r_3)=(27,5)$, and $(n_4,r_4)=(81,7)$. Then our construction (with designed distances found from \eqref{eq:d}) gives rise to a strongly optimal H-LRC code of length $n=n_4$ and dimension $k=r_4$ with local distances $\delta_2=6,\delta_3=17$ and distance $d=\delta_4=53$ over a finite field $\mathbb{F}_q$ where $81|(q-1)$ (for example, we can take $q=163$).
\end{example}

	\subsection{Hierarchical cyclic codes of unbounded length}
In this section we construct a family of H-LRC codes with distance $d=\delta_h+1,h\ge 1$ and unbounded length. The construction combines the idea of \cite{luo2018optimal} with H-LRC codes of the previous section.

Let $1 < r_1 < r_2 < \ldots < r_h$ be integers. Let $1=\delta_0 < \delta_1$ and let $\delta_2,\delta_3,\ldots,\delta_h$ be as in \eqref{eq:d}. 
Again, we put $n_1=r_1+\delta_1-\delta_0$ and let $n_{i+1}=\nu_i n_i$ for $1\leq i\leq h-1,$ where $\nu_i\geq \lceil r_{i+1}/r_i \rceil$ is an integer. Let $\mathbb{F}_{q^m},m\ge 1$ be a finite field and let $n_h| (q-1)$. Let $n = q^m-1$ and observe that $n_h| n$. 
Let $\alpha\in \mathbb{F}_{q^m}$ be a primitive $n$-th root of unity.
Let $\cZ_h$ be constructed by the procedure in \eqref{eq:def-set-1} and \eqref{eq:def-set-2}. Finally, define
\begin{align}
\cL = \bigcup_{s=0}^{n/n_h-1}(\cZ_h+sn_h),\quad
\cZ = \cL \cup \{0\}.\label{eq:long}
\end{align}

Consider a cyclic code $\cC$ with generator polynomial
\begin{align}
\label{eq:c1}
g(x) = \prod_{t\in\cZ}(x-\alpha^t).
\end{align}
As is easily seen, $g(x)\in \mathbb{F}_q[x]$. Indeed, 
\begin{align*}
g(x) &= (x-1)\prod_{t\in\cL}(x-\alpha^t)\\
&= (x-1)\prod_{s=0}^{n/n_h-1}\prod_{t\in\cZ_h}(x-\alpha^{sn_h+t})\\
&= (x-1)\prod_{t\in\cZ_h}(x^{n/n_h}-\alpha^{nt/n_h}).
\end{align*}
For the last equality we note that $x^{n/n_h}-\alpha^{nt/n_h}=\prod_{s=0}^{n/n_h-1}(x-\alpha^{sn_h+t}).$
Observe that for $t\in\cZ_h$ we have $(\alpha^{nt/n_h})^{q-1}=1$ since $n_h|(q-1)$.
It follows that $\alpha^{nt/n_h}\in\mathbb{F}_q$ for $t\in\cZ_h$ and thus $g(x)\in \mathbb{F}_q[x]$.

\begin{proposition}
	\label{cl:d1}
	Let $\cC=\langle g(x)\rangle\in \mathbb{F}_q[x]/(x^n-1)$ be a cyclic code. Then $\dim(\cC)=nr_n/n_h-1$, $d(\cC)\geq \delta_h+1$, and 
the locality parameters are $(r_i,\delta_i)$ for $1\leq i\leq h$.
\end{proposition}

\begin{proof}
The dimension of $\cC$	is found as 
   $$ 
   k=n-\deg(g(x))=n-1-\frac{|\cZ_h|n}{n_h}=
	 \frac{nr_h}{n_h}-1,
	 $$
where the last equality follows by Proposition~\ref{cl:card}. 
The distance of the code $\cC$ is $d\geq \delta_h+1$ since $g(x)$ has consecutive roots $\alpha^t,t=0,\ldots,\delta_h-1$. 
	
	The locality parameters of the code $\cC$ follow immediately by Lemma~\ref{le:local0} and Proposition~\ref{cl:card}.
\end{proof}

The next lemma provides the conditions when the code $\cC$ is optimal. Its proof amounts to a calculation based on Lemma~\ref{le:opt} and
Proposition~\ref{cl:d1}.

\begin{lemma}
	\label{le:d1-opt}
	Suppose that for $1\leq i\leq h-1$, conditions \eqref{eq:opt} are satisfied. Further, suppose that 
	\begin{align}
	\frac{n}{n_h}\left\lceil \frac{r_h}{r_l} \right\rceil = \left\lceil \frac{k}{r_l} \right\rceil, \quad 1\leq l\leq h-1.\label{eq:opt-d1}
	\end{align} Then the code $\cC$ is optimal.
\end{lemma}
\begin{proof}
By Lemma~\ref{le:opt}, we have
\begin{align}
\delta_{h} = n_{h}-r_{h}+\delta_{h-1}-\sum_{l=1}^{h-1}\left\lceil \frac{r_{h}}{r_l} \right\rceil (\delta_l-\delta_{l-1}).\label{eq:d1}
\end{align}

By Proposition~\ref{cl:d1}, we have $k=nr_h/n_h-1$. Using the bound \eqref{eq:sb}, the distance of the code cannot exceed
\begin{align}
n-k+\delta_{h}-&\sum_{l=1}^{h}\left\lceil \frac{k}{r_l} \right\rceil (\delta_l-\delta_{l-1})\nonumber\\
&= n-\frac{nr_h}{n_h}+1+\delta_h-\frac{n}{n_h} (\delta_h-\delta_{h-1})-\sum_{l=1}^{h-1}\left\lceil \frac{k}{r_l} \right\rceil (\delta_l-\delta_{l-1})\label{eq:k}\\
&= 1+\delta_h+\frac{n}{n_h}\sum_{l=1}^{h-1}\left\lceil \frac{r_{h}}{r_l} \right\rceil (\delta_l-\delta_{l-1})-\sum_{l=1}^{h-1}\left\lceil \frac{k}{r_l} \right\rceil (\delta_l-\delta_{l-1})\label{eq:d1-sb1}\\
&= 1+\delta_h,\label{eq:d1-sb2}
\end{align}
where \eqref{eq:k} follows since $r_h>1$ implies $\lceil k/r_h \rceil=n/n_h$, in \eqref{eq:d1-sb1} we used \eqref{eq:d1}, and \eqref{eq:d1-sb2} follows by \eqref{eq:opt-d1}.
Hence, the code $\cC$ has the largest possible distance $d=\delta_h+1$.
\end{proof}

In particular, the conditions in Lemma~\ref{le:d1-opt} are satisfied when $r_i|r_{i+1}, i=1,\dots,h-1$ and $r_h|k.$

As in Sec.~\ref{sec:hlrc}, the code $\cC$ constructed above in this section has the strong optimality property
if the optimality conditions in Lemma~\ref{le:d1-opt} are satisfied. 
Specifically, the main difference between the construction in this section and 
the one in the previous section is in the final step of constructing the defining set $\cZ$, which also includes element $0$ (i.e., $\alpha^0$). 
By an argument similar to Sec.~\ref{sec:hlrc}, one can show $n-n_i+t_i\notin \cZ$ for 
every $t_i\in (\{0,\ldots,n_i-1\}\setminus\cZ_i)+(n-n_i), 1\leq i\leq h$ and so strong optimality follows.

\begin{example}
	Consider the case $h=3$. Let $r_1=2$ and $\delta_1=2$. Then $n_1=3$. Let $(n_2,r_2)=(9,3)$ and $(n_3,r_3)=(27,5)$.
	Let $m\geq 1$ be an arbitrary integer and $q=163$.
	Then our construction (with designed local distances given by \eqref{eq:d}) gives rise to a strongly optimal H-LRC code of length $n=163^m-1$ and dimension $k=2(163^m-1)/9$ with local distances $\delta_2=6,\delta_3=17$ and distance $d=18$ over $\mathbb{F}_q$.
\end{example}

Recall that \cite{beemer18} shows that the length of an optimal LRC code in the general case cannot be greater than a certain power of the
alphabet size $q$. Using similar arguments, it might be possible to derive upper bounds on the length of optimal H-LRC codes in the general case; however already in the case of $(r,\delta)$ locality addressed in \cite{CaiMiaoSchwartzTang18} (with just a single level of hierarchy), following this route requires cumbersome calculations.


\section{Convolutional codes with locality}\label{sec:cc}

It has been recognized a long while ago that quasi-cyclic codes can be encoded convolutionally, and multiple papers 
constructed families of convolutional codes from their quasicyclic counterparts \cite{Justesen1990,esmaeili1998link,Tanner2004}. In this section, we present a family of convolutional codes with locality by relying on the tailbiting version of convolutional codes \cite{SvT79}. 
We single out this approach because it enables us to establish the locality properties of convolutional codes based on the
properties of cyclic H-LRC codes constructed above in this paper.

We begin with a brief reminder of the basic notions for convolutional codes \cite{JZ15}. 
Let $D$ be an indeterminate and let $\ff_q(D)$ be the field of rational functions of one variable over $\ff_q.$ A $q$-ary $(n,k)$ convolutional code
$\cC$ is a linear $k$-dimensional subspace of $\ff_q(D)^n.$ A generator matrix $G(D)=(g_{ij}(D))$ of the code $\cC$ 
is a $k\times n$ matrix with entries in $\ff_q(D)$ whose rows form a basis of $\cC$.  Thus, the code $\cC$ is a linear space 
$\{u(D)G(D) \mid u(D)\in \ff_q(D)^k\}.$ The matrix $G(D)$ can be transformed to the polynomial form by 
multiplying every element by the least common denominator of its entries. The transformed matrix generates the same code $\cC$, and so 
in the sequel we will consider only \emph{polynomial generator matrices}. Below we will 
assume that the generator matrix $G(D)$ is a $k\times n$ matrix with entries in $\ff_q[D],$ where $\ff_q[D]$ is the ring of polynomials over $\ff_q.$

For $1\leq i\leq k$, the \emph{degree $m_i$ of the $i$-th row of $G(D)$} is the maximum degree of the entries in row $i$, namely, $m_i=\max_{1\leq j\leq n}\deg(g_{i,j}(D))$. As with linear block codes, the encoding of a convolutional code depends on the choice of a generator matrix. The maximum degree $M:=\max_{1\le i\le k}m_i$ is called the \emph{memory} of the encoder. The generator matrix of the code $\cC$ can also
be written in the form
   \begin{equation}\label{eq:G}
   G=\begin{bmatrix} G_0&G_1&\dots&G_M\\&G_0&G_1&\dots&G_M\\&&\ddots&\ddots&&\ddots\end{bmatrix},
   \end{equation}
where each $G_i$ is a $k\times n$ matrix over $\ff_q.$ The codeword of the code $\cC$ is obtained as a product $uG,$ where $u$ is a semi-infinite input sequence of symbols of $\ff_q.$

With a given convolutional code $\cC$ one can associate a multitude of distance measures. In direct analogy with block codes, one defines the
{\em free distance} of the code $\cC$ as the minimum Hamming weight of the Laurent expansions of the nonzero codewords. 

Another distance measure of interest is the so-called \emph{column distance} of the code \cite[p.162]{JZ15}. 
To define it, let $\cC_{[0,j]}$ be the truncated code of $\cC$ at the $j$-th time instant, $j\ge 0,$ namely, 
    $$
   \cC_{[0,j]}=\{c_{[0,j]}(D)=\sum_{i=0}^{j}c_iD^i\mid c(D)=\sum_{i\geq 0}c_iD^i\in\cC \}.
    $$
This is a linear block code of length $n(j+1)$, and by \eqref{eq:G} its generator matrix can be written in the form
  \begin{equation}\label{eq:Gj}
    G_j^c=\begin{bmatrix} G_0&G_1&\dots &G_j\\&G_0&\dots&G_{j-1}\\&&\ddots&\vdots\\&&&G_0\end{bmatrix}
  \end{equation}
where we put $G_l=0$ for $l> M.$ Clearly, the code $\cC_{[0,j]}$ is obtained by truncating the code $\cC$ to its first $j+1$ entries.
A codeword of $\cC_{[0,j]}$ has the form $(c_0,c_1,\dots,c_j),$ where for $j\le M$ and each $l=0,1,\dots,j$
   \begin{equation}\label{eq:tc}
   c_l=\sum_{i=0}^l u_i G_{l-i},
   \end{equation}
   where $c_l=(c_l^{(1)},\dots,c_l^{(n)})$ for each $l$.
   
We assume throughout that $G_0$ has full rank, so the mapping $\ff_q^k\stackrel{G}\to\ff_q^n$ given by $u_0 G_0\mapsto c_0$ is injective. 

\begin{definition}\label{def:djc}
	 For $j\geq 0$ the $j$-th column distance of $\cC$ is given by
	\begin{align*}
		d_j^c=\min\{\mathrm{wt}(c_{[0,j]}(D)) \mid c_{[0,j]}(D)\in\cC_{[0,j]},c_0\neq 0 \}.
	\end{align*}
\end{definition}
Clearly, the value of $d_j^c$ is at least the minimum distance of the truncated code $\cC_{[0,j]}.$ This follows because
for the column distance we seek the minimum of pairwise distances of codewords that differ in the first coordinate, while the
standard minimum distance computation does not involve this assumption. In many cases the column distance is in fact strictly greater. This remark 
is important for the sliding window repair which enables one to correct more erasures than would be possible for block codes relying on their minimum
distance.

Convolutional codes support several forms of erasure repair. One of them, called the sliding window repair \cite{TRV2012,MN19}, is based on the 
column distance and is used to correct erasures in streaming applications \cite{TRV2012}. We illustrate the idea of sliding window repair in Fig.~
\ref{fig:swr}, representing a code sequence of the code $\cC$ as a semi-infinite matrix whose columns are length $n$ vectors $c_l, l\ge 0$, and whose 
row $c^{(i)},i=1,\dots,n$ represents the stream formed by the $i$th coordinates of the symbols $c_l,l=0,1,\dots.$ We begin with fixing $j$ based on 
the value of the column distance $d_j^c$ of the code. The box in the figure shown with dashed lines represents the window of length $j+1$ that 
contains the truncated code at time $l\ge j.$ The erasures within the sliding window can clearly be repaired as long as their number at no point exceeds $d_j^c-1.$

Having in mind streaming applications, one may argue that a more efficient way of repairing erasures is to rely either on the symbols at a fixed time instant, or on a small group of symbols contained within the same stream $i.$ Accordingly, in the next two subsections we define two types of locality for convolutional codes, calling them the column and row localities.

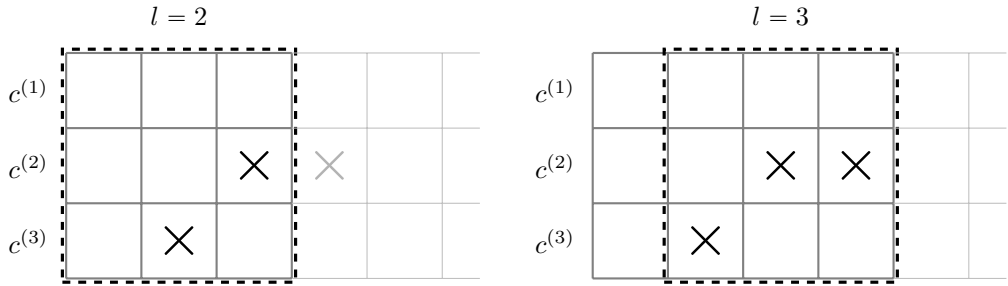
\begin{figure}[ht]
\hspace*{.4in}	\begin{tikzpicture}
	\node at (1.5,7.5) {$l=2$};
	\node at (-0.5,6.5) {$c^{(1)}$};
	\node at (-0.5,5.5) {$c^{(2)}$};
	\node at (-0.5,4.5) {$c^{(3)}$};
	\draw[step=1cm,gray,thick] (0,4) grid (3,7);
	\node at (1.5,4.5) {\Huge $\times$};
	\node at (2.5,5.5) {\Huge $\times$};
	\node [opacity=0.3] at (3.5,5.5) {\Huge $\times$};
	\draw[very thick,dashed] (-.05,3.95) rectangle (3.05,7.05);
	\draw[step=1cm,gray!50,thin] (3.001,4) grid (5.5,7);
	\node at (9.5,7.5) {$l=3$};
	\node at (6.5,6.5) {$c^{(1)}$};
	\node at (6.5,5.5) {$c^{(2)}$};
	\node at (6.5,4.5) {$c^{(3)}$};
	\draw[step=1cm,gray,thick] (7,4) grid (11,7);
	\node at (8.5,4.5) {\Huge $\times$};
	\node at (9.5,5.5) {\Huge $\times$};
	\node at (10.5,5.5) {\Huge $\times$};
	\node at (4.5,2.5) {};
	\draw[very thick,dashed] (7.95,3.95) rectangle (11.05,7.05);
	\draw[step=1cm,gray!50,thin] (11.001,4) grid (12.5,7);
	\end{tikzpicture}
\vspace*{-.5in}	\caption{{\sc Sliding window repair.} Suppose $\cC$ is an $(3,2)$ convolutional code with $d^c_2=4$. Let 
$(c^{(1)},c^{(2)},c^{(3)})\in\cC$ be a codeword, where the crosses denote erasures. At time instant $l=2$, there are two erasures in the window of length three (dashed box in the left figure), which is less than the $2$nd column distance. However, neither of the two erasures are in the first column of the window, and thus their recovery is postponed until later. At time $l=3$ (right figure), the sliding window contains 3 erasures, of which one is in the first column. This erasure can be recovered from the other symbols in 
the window. The remaining erasures are corrected in the next steps as long as the number of erasures in the window
does not exceed $d^c_2-1.$
}
	\label{fig:swr}
\end{figure}

\subsection{Convolutional codes with column locality}

Column locality was introduced in \cite{MN19}. First let us define the {\em $i$th column code} $\cC_i,i\ge 0$ of a convolutional code $\cC$ as a block 
code of length $n$ given by
   $$
   \cC_i=\{c_i\mid c(D)\in \cC \}.
   $$
We say that a convolutional code $\cC$ has $(r,\delta)$ column locality if for all $i$ the codes $\cC_i$ have the $(r,\delta)$ locality
property.

The results in \cite{MN19} are based on a version of this definition that requires that {\em only the code $\cC_M$} support $(r,\delta)$ locality. 
This restriction may seem too narrow until one realizes that if locality is present in the code $\cC_M$, then every code $\cC_i,i\ge 0$ has the $(r,
\delta)$ locality property. This follows immediately from \eqref{eq:tc} and the definition of $\cC_i$ because $\cC_i=\cC_M$ for $i>M$ and $\cC_i$ 
forms a linear subcode of $\cC_M$ otherwise. The only difference between this definition and the one given above is that under the approach of 
\cite{MN19}, every code $\cC_i$ has similarly aligned repair groups which are propagated from the repair groups of $\cC_M$, while our definition 
allows differently aligned repair groups for different values of $i$.

To enable local repair, we simply assume that every column of the codeword forms a block code with $(r,\delta)$ locality.
An example is given in Fig~\ref{fig:swr-c}, demonstrating sliding window repair combined with column locality.

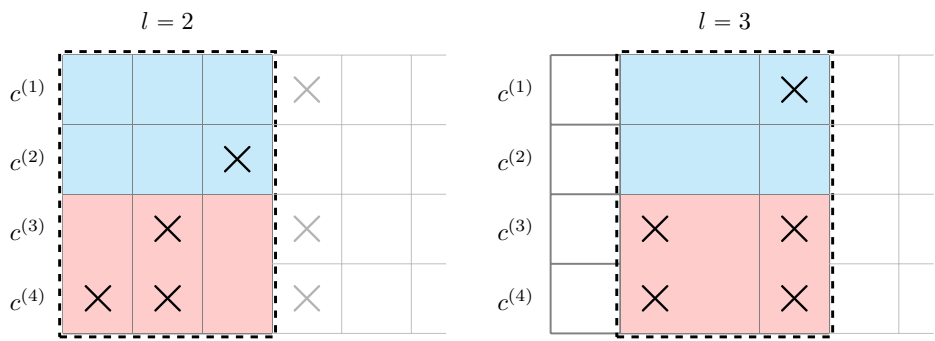
\begin{figure}[h]
	\centering
\resizebox{.8\textwidth}{!}	{\begin{tikzpicture}[
rednode/.style={rectangle, fill=red!20, minimum size=9.8mm},
cyannode/.style={rectangle, fill=cyan!20, minimum size=9.8mm},
]
\node at (1.5, 9.5) {$l=2$};
\node at (-0.5,8.5) {$c^{(1)}$};
\node at (-0.5,7.5) {$c^{(2)}$};
\node at (-0.5,6.5) {$c^{(3)}$};
\node at (-0.5,5.5) {$c^{(4)}$};
\draw[step=1cm,gray,thick] (0,5) grid (3,9);
\node[rednode] at (0.5,5.5) {\Huge $\times$};
\node[rednode] at (0.5,6.5) {};
\node[cyannode] at (0.5,7.5) {};
\node[cyannode] at (0.5,8.5) {};
\node[cyannode] at (1.5,8.5) {};
\node[cyannode] at (1.5,7.5) {};
\node[rednode] at (1.5,6.5) {\Huge $\times$};
\node[rednode] at (1.5,5.5) {\Huge $\times$};
\node[cyannode] at (2.5,8.5) {};
\node[cyannode] at (2.5,7.5) {\Huge $\times$};
\node[rednode] at (2.5,6.5) {};
\node[rednode] at (2.5,5.5) {};
\node[opacity=0.3] at (3.5,8.5) {\Huge $\times$};
\node[opacity=0.3] at (3.5,6.5) {\Huge $\times$};
\node[opacity=0.3] at (3.5,5.5) {\Huge $\times$};
\draw[very thick,dashed] (-0.05,4.95) rectangle (3.05,9.05);
\draw[step=1cm,gray!50,thin] (3.001,5) grid (5.5,9);
\node at (9.5,9.5) {$l=3$};
\node at (6.5,8.5) {$c^{(1)}$};
\node at (6.5,7.5) {$c^{(2)}$};
\node at (6.5,6.5) {$c^{(3)}$};
\node at (6.5,5.5) {$c^{(4)}$};
\draw[step=1cm,gray,thick] (7,5) grid (11,9);
\node[cyannode] at (8.5,8.5) {};
\node[cyannode] at (8.5,7.5) {};
\node[rednode] at (8.5,6.5) {\Huge $\times$};
\node[rednode] at (8.5,5.5) {\Huge $\times$};
\node[cyannode] at (9.5,8.5) {};
\node[cyannode] at (9.5,7.5) {};
\node[rednode] at (9.5,6.5) {};
\node[rednode] at (9.5,5.5) {};
\node[cyannode] at (10.5,8.5) {\Huge $\times$};
\node[cyannode] at (10.5,7.5) {};
\node[rednode] at (10.5,6.5) {\Huge $\times$};
\node[rednode] at (10.5,5.5) {\Huge $\times$};
\draw[very thick,dashed] (7.95,4.95) rectangle (11.05,9.05);
\draw[step=1cm,gray!50,thin] (11.001,5) grid (12.5,9);
\end{tikzpicture}}
	\caption{{\sc Sliding window repair with column locality.} Suppose $\cC$ is a $(4,2)$ convolutional code with $(1,2)$ column locality and $d_2^c=6.$ 
Let $(c^{(1)},c^{(2)},c^{(3)},c^{(4)})\in\cC$ be a codeword, where the crosses denote erasures, and different repair groups in the
columns are shown in different colors. At time $l=2$, the window of length $j+1=3$ contains 4 erasures. Of these, the symbols $c^{(4)}_0,c^{(2)}_2$
can be recovered within their repair groups. Then for $l=3$ the window contains 5 erasures, of which two in the first column can be repaired from the other symbols in the window, while the symbol $c^{(1)}_3$ can be recovered locally. }
	\label{fig:swr-c}
\end{figure}

\subsection{Convolutional codes with row locality}
In this section we introduce and study another notion of locality for convolutional codes. Given a convolutional code $\cC,$ define the
{\em $i$th row code} $\cC^{(i)}_{[0,j]},1\leq i\leq n$ truncated at $j$th time instant, $j\ge 0,$ as follows:
  \begin{equation}\label{eq:row code}
  \cC^{(i)}_{[0,j]}=\{c^{(i)}_{[0,j]}=(c^{(i)}_0,\ldots,c^{(i)}_j) \mid c\in \cC \}.
  \end{equation}
We say that $\cC$ has $(r,\delta)$ row locality if for all $t\geq 0$ the codes $\cC_{[t,t+j]}^{(i)},1\leq i\leq n$ have the $(r,\delta)$ locality property, where $j\geq 0$ is fixed. 

In the case of tailbiting codes, it is more convenient to give this definition in the following form, which will also be used in our constructions
below.
\begin{definition}
	Let $j\ge0$ be fixed.
A convolutional code $\cC$ has $(r,\delta)$ row locality at time $j$ if every code $\cC^{(i)}_{[0,j]},1\leq i\leq n$ has $(r,\delta)$ locality.
\end{definition}

We give two examples of repair with row locality, see Figures~\ref{fig:swr-r} and \ref{fig:swr-r-tb}. The first of them illustrates the above
definition while the second specializes it to tailbiting codes. {Let us stress that whenever local repair by rows is not possible,
we fall back on sliding window repair relying on the column distance of the truncated code.}

\begin{figure}
	\centering
\resizebox{.9\textwidth}{!}{	\begin{tikzpicture}[
rednode/.style={rectangle, fill=red!20, minimum size=9.8mm},
cyannode/.style={rectangle, fill=cyan!20, minimum size=9.8mm},
]
\node at (2,9.5) {$l=3$};
\node at (-0.5,8.5) {$c^{(1)}$};
\node at (-0.5,7.5) {$c^{(2)}$};
\node at (-0.5,6.5) {$c^{(3)}$};
\node at (-0.5,5.5) {$c^{(4)}$};
\draw[step=1cm,gray,thick] (0,5) grid (4,9);
\node[rednode] at (0.5,5.5) {};
\node[cyannode] at (1.5,5.5) {\Huge $\times$};
\node[rednode] at (2.5,5.5) {\Huge $\times$};
\node[cyannode] at (3.5,5.5) {};
\node[rednode] at (0.5,6.5) {};
\node[cyannode] at (1.5,6.5) {\Huge $\times$};
\node[rednode] at (2.5,6.5) {};
\node[cyannode] at (3.5,6.5) {\Huge $\times$};
\node[rednode] at (0.5,7.5) {\Huge $\times$};
\node[cyannode] at (1.5,7.5) {};
\node[rednode] at (2.5,7.5) {};
\node[cyannode] at (3.5,7.5) {\Huge $\times$};
\node[rednode] at (0.5,8.5) {};
\node[cyannode] at (1.5,8.5) {\Huge $\times$};
\node[rednode] at (2.5,8.5) {};
\node[cyannode] at (3.5,8.5) {\Huge $\times$};
\draw[very thick,dashed] (-0.05,4.95) rectangle (4.05,9.05);
\draw[step=1cm,gray!50,thin] (4.001,5) grid (6.5,9);
\node at (11,9.5) {$l=4$};
\node at (7.5,8.5) {$c^{(1)}$};
\node at (7.5,7.5) {$c^{(2)}$};
\node at (7.5,6.5) {$c^{(3)}$};
\node at (7.5,5.5) {$c^{(4)}$};
\draw[step=1cm,gray,thick] (8,5) grid (13,9);
\node[rednode] at (9.5,5.5) {};
\node[cyannode] at (10.5,5.5) {};
\node[rednode] at (11.5,5.5) {};
\node[cyannode] at (12.5,5.5) {};
\node[rednode] at (9.5,6.5) {\Huge $\times$};
\node[cyannode] at (10.5,6.5) {};
\node[rednode] at (11.5,6.5) {\Huge $\times$};
\node[cyannode] at (12.5,6.5) {};
\node[rednode] at (9.5,7.5) {};
\node[cyannode] at (10.5,7.5) {};
\node[rednode] at (11.5,7.5) {};
\node[cyannode] at (12.5,7.5) {};
\node[rednode] at (9.5,8.5) {\Huge $\times$};
\node[cyannode] at (10.5,8.5) {};
\node[rednode] at (11.5,8.5) {\Huge $\times$};
\node[cyannode] at (12.5,8.5) {};
\draw[very thick,dashed] (8.95,4.95) rectangle (13.05,9.05);
\draw[step=1cm,gray!50,thin] (13.001,5) grid (14.5,9);
\end{tikzpicture}}
	\caption{{\sc Sliding window repair with row locality.} Suppose $\cC$ is a $(4,2)$ convolutional code with $(1,2)$ row locality 
$d_3^c=8.$	Let $(c^{(1)},c^{(2)},c^{(3)},c^{(4)})\in\cC$ be a codeword, where the crosses denote erasures, and different repair groups in the
rows are shown in different colors. At time $l=3,$ by row locality, the symbols $c^{(2)}_0,c^{(2)}_3,c^{(4)}_1,c^{(4)}_2$ can be recovered from the other symbols in their respective repair groups. At time $l=4$ there are four erasures in the window of length four, which is smaller than $d_3^c,$ so they are recoverable. Thereafter the two remaining erasures can be recovered relying on row locality.}
	\label{fig:swr-r}
\end{figure}
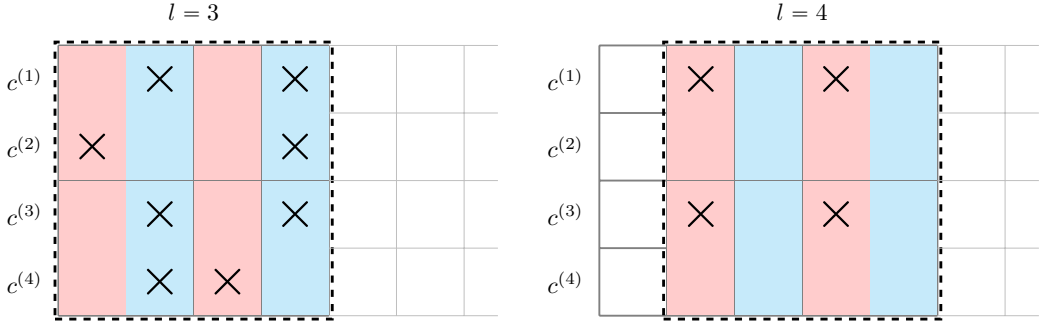

\begin{figure}
\resizebox{\textwidth}{!}{
	\hspace*{-.2in}\begin{tikzpicture}[
	rednode/.style={rectangle, fill=red!20, minimum size=9.8mm},
	greennode/.style={rectangle, fill=green!20, minimum size=9.8mm},
	cyannode/.style={rectangle, fill=cyan!20, minimum size=9.8mm},
	yellownode/.style={rectangle, fill=yellow!20, minimum size=9.8mm}, 
	]
	\node at (4,9.5) {Sliding window starting from time $l=0$};
	\node at (-0.5,8.5) {$c^{(1)}$};
	\node at (-0.5,7.5) {$c^{(2)}$};
	\node at (-0.5,6.5) {$c^{(3)}$};
	\node at (-0.5,5.5) {$c^{(4)}$};
	\draw[step=1cm,gray,thick] (0,5) grid (8,9);
	\node[rednode] at (0.5,5.5) {\Huge $\times$};
	\node[cyannode] at (1.5,5.5) {};
	\node[greennode] at (2.5,5.5) {\Huge $\times$};
	\node[yellownode] at (3.5,5.5) {};
	\node[rednode] at (4.5,5.5) {};
	\node[cyannode] at (5.5,5.5) {\Huge $\times$};
	\node[greennode] at (6.5,5.5) {\Huge $\times$};
	\node[yellownode] at (7.5,5.5) {};
	\node[rednode] at (0.5,6.5) {\Huge $\times$};
	\node[cyannode] at (1.5,6.5) {};
	\node[greennode] at (2.5,6.5) {\Huge $\times$};
	\node[yellownode] at (3.5,6.5) {};
	\node[rednode] at (4.5,6.5) {\Huge $\times$};
	\node[cyannode] at (5.5,6.5) {};
	\node[greennode] at (6.5,6.5) {\Huge $\times$};
	\node[yellownode] at (7.5,6.5) {};
	\node[rednode] at (0.5,7.5) {\Huge $\times$};
	\node[cyannode] at (1.5,7.5) {};
	\node[greennode] at (2.5,7.5) {\Huge $\times$};
	\node[yellownode] at (3.5,7.5) {};
	\node[rednode] at (4.5,7.5) {};
	\node[cyannode] at (5.5,7.5) {\Huge $\times$};
	\node[greennode] at (6.5,7.5) {\Huge $\times$};
	\node[yellownode] at (7.5,7.5) {};
	\node[rednode] at (0.5,8.5) {\Huge $\times$};
	\node[cyannode] at (1.5,8.5) {};
	\node[greennode] at (2.5,8.5) {\Huge $\times$};
	\node[yellownode] at (3.5,8.5) {};
	\node[rednode] at (4.5,8.5) {\Huge $\times$};
	\node[cyannode] at (5.5,8.5) {};
	\node[greennode] at (6.5,8.5) {\Huge $\times$};
	\node[yellownode] at (7.5,8.5) {};
	\draw[very thick,dashed] (-0.05,4.95) rectangle (8.05,9.05);
	\node at (13,9.5) {Sliding window starting from time $l=2$};
	\node at (8.7,8.5) {$c^{(1)}$};
	\node at (8.7,7.5) {$c^{(2)}$};
	\node at (8.7,6.5) {$c^{(3)}$};
	\node at (8.7,5.5) {$c^{(4)}$};
	\draw[step=1cm,gray,thick] (9,5) grid (17,9);
	\node[rednode] at (9.5,5.5) {};
	\node[cyannode] at (10.5,5.5) {};
	\node[greennode] at (11.5,5.5) {\Huge $\times$};
	\node[yellownode] at (12.5,5.5) {};
	\node[rednode] at (13.5,5.5) {};
	\node[cyannode] at (14.5,5.5) {};
	\node[greennode] at (15.5,5.5) {\Huge $\times$};
	\node[yellownode] at (16.5,5.5) {};
	\node[rednode] at (9.5,6.5) {};
	\node[cyannode] at (10.5,6.5) {};
	\node[greennode] at (11.5,6.5) {\Huge $\times$};
	\node[yellownode] at (12.5,6.5) {};
	\node[rednode] at (13.5,6.5) {};
	\node[cyannode] at (14.5,6.5) {};
	\node[greennode] at (15.5,6.5) {\Huge $\times$};
	\node[yellownode] at (16.5,6.5) {};
	\node[rednode] at (9.5,7.5) {};
	\node[cyannode] at (10.5,7.5) {};
	\node[greennode] at (11.5,7.5) {\Huge $\times$};
	\node[yellownode] at (12.5,7.5) {};
	\node[rednode] at (13.5,7.5) {};
	\node[cyannode] at (14.5,7.5) {};
	\node[greennode] at (15.5,7.5) {\Huge $\times$};
	\node[yellownode] at (16.5,7.5) {};
	\node[rednode] at (9.5,8.5) {};
	\node[cyannode] at (10.5,8.5) {};
	\node[greennode] at (11.5,8.5) {\Huge $\times$};
	\node[yellownode] at (12.5,8.5) {};
	\node[rednode] at (13.5,8.5) {};
	\node[cyannode] at (14.5,8.5) {};
	\node[greennode] at (15.5,8.5) {\Huge $\times$};
	\node[yellownode] at (16.5,8.5) {};
	\draw[very thick,dashed] (17.05,9.05) -- (11.05,9.05) -- (11.05,4.95) -- (17.05,4.95);
	\draw[very thick,dashed] (8.95,9.05) -- (10.95,9.05) -- (10.95,4.95) -- (8.95,4.95);
	\end{tikzpicture}}
	\caption{{\sc Sliding window repair with row locality for tailbiting codes.} Suppose $\cC$ is a $(4,2)$ unit memory tailbiting convolutional code
with $(1,2)$ row locality and $d_7^c=16$. Let $(c^{(1)},c^{(2)},c^{(3)},c^{(4)})\in\cC$ be a codeword, where the crosses denote erasures, and 
different repair groups in the rows are shown in different colors. 
There are a total of 16 erasures. First we engage row locality to repair symbols $c^{(2)}_0,c^{(2)}_5,c^{(4)}_0,c^{(4)}_5,$ 
whereupon 12 erasures remain. In the sliding window starting from time $l=0$, the remaining two erasures in the first column can be recovered. After that, $c^{(1)}_4,c^{(3)}_4$ can be recovered locally. \\
\hspace*{.3in}Next, we move the sliding window to start at time $l=2.$ (Note that the window is wrapped around.) The erasures in the first column of this window can be recovered since the number of erasures is smaller than the column distance, and after that the remaining erasures
can be recovered relying on row locality.}
	\label{fig:swr-r-tb}
\end{figure}
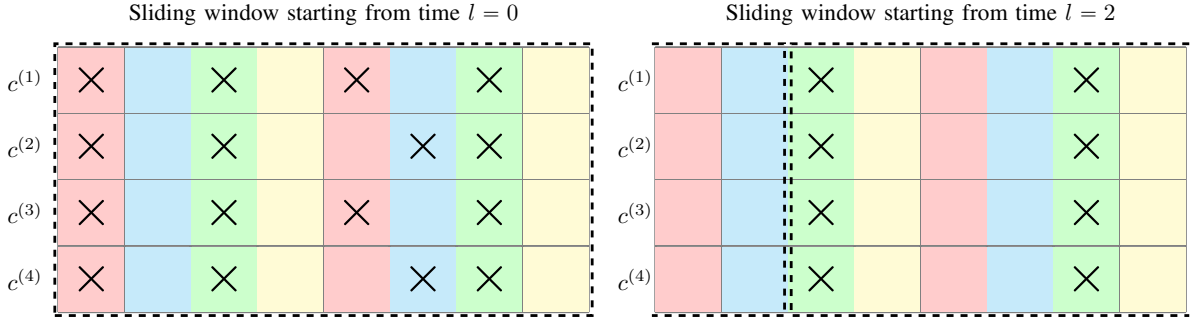

The problem that we address is to construct convolutional codes with locality and large column distance. This is similar to the problem
studied in \cite{MN19} and also to the case of block codes with locality and large minimum distance.pursuing constructions with large column distance and $(r,\delta)$ locality. 

We begin with deriving an upper bound on the column distance of the truncated code with either 
column or row locality property. 

\begin{proposition}
	\label{prop:lrcc-max}
(a)	Let $\cC$ be an $(n,k)$ convolutional code with $(r,\delta)$ (column or row) locality. Then for any $j\geq 0$, the $j$-th column distance satisfies
	\begin{align}
		d_j^c\leq (n-k)(j+1)+\delta-\left\lceil\frac{k}{r}\right\rceil(\delta-1).\label{eq:sb-lrcc}
	\end{align} 

(b)
	Equality in \eqref{eq:sb-lrcc} implies that for all $i\leq j$, the $i$-th column distance satisfies
	\begin{align}
		d_i^c= (n-k)(i+1)+\delta-\left\lceil\frac{k}{r}\right\rceil(\delta-1).\label{eq:sb-lrcc1}
	\end{align}
\end{proposition}
The proof is given in Appendix~\ref{app:proof-of-lrcc-max}.

Part (b) of this proposition is similar to the propagation of the column distance optimality property in the case of general convolutional codes proved in \cite{gluesing2006strongly}. Namely, the Singleton bound implies that the column distance for all $j$ satisfies
  \begin{equation}\label{eq:sbc}
  d_j^c\le (n-k)(j+1)+1,
  \end{equation}
and equality for a given $j$ implies that all the other distances $d_i^c, i\le j$ also attain their versions of the Singleton bound with equality.

\subsection{Convolutional codes and quasicyclic codes}\label{sec:cc-qc}
A transformation between these two code families was constructed in \cite{SvT79}, and it has led to a broader family of convolutional codes and trellises now known as {\em tailbiting codes} (tailbiting trellises) \cite{JZ15}. 
An $(n(m+1),k(m+1))$ quasicyclic code can be defined by a generator matrix
   $$
   G=(G_{ij}),\;i=0,\dots,k-1, j=0,\dots,n-1
   $$
where each $G_{ij}$ is an $(m+1)\times (m+1)$ circulant matrix. 
With a given matrix $G_{ij}$ we associate a polynomial 
$g_{ij}(D)=\sum_{l=0}^{m} g_{l}D^{l},$ where $(g_0,\dots,g_{m})$ is the first row
of the matrix. Then the $k\times n$ generator matrix $G(D) = (g_{ij}(D))$ defines an $(n,k)$ convolutional code. 
The authors of \cite{SvT79} showed that if one takes the input sequences of the convolutional code in the form 
  \begin{align}
	u(D)=\sum_{l=-M}^mu_lD^l \text{ such that } u_{-s} = u_{m+1-s} \text{ for } s = 1,\ldots,M, \label{eq:qc}
  \end{align}
then the convolutional code is equivalent to the quasicyclic code defined above. In other words, the quasicyclic code 
can be encoded convolutionally, and the convolutional code with ``symmetric'' input sequences as in \eqref{eq:qc} is exactly the quasicyclic code.

\subsection{A family of tailbiting convolutional codes with row locality}
We come to the main result of this section, which is a construction of a family of convolutional codes with $(r,\delta)$ row locality. 
This family of codes has the largest possible minimum distance for the truncated code $\cC_{[0,j]},$ however we stop short of showing that 
the $j$-th column distance attains \eqref{eq:sb-lrcc} with equality. 
The construction is achieved by exploiting a connection between quasi-cyclic codes and convolutional codes discussed above. In high-level terms,
our plan is to construct a cyclic code from its set of zeros chosen according to the procedure in Sec.~\ref{sec:h}, writing it in a 
quasicyclic form (via a circulant generator matrix) and to construct a convolutional code using the technique discussed above.

\vspace*{.05in} Let us first construct an $(n(j+1),k(j+1))$ cyclic LRC codes with $(r,\delta)$ locality. We proceed similarly to Sec.~\ref{sec:hlrc}. We will need a few assumptions regarding the parameters of the code. Let $j\geq 0$ be such that $k\leq j+1\leq n$ and that $j+1=(r+\delta-1)\nu$ where $\nu\geq 1$. Let $\ff_q$ be a finite field such that $n(j+1)\mid (q-1)$ 
and let $\alpha\in\ff_q$ be a primitive root of unity of order $n(j+1)$ in $\ff_q.$

The set of zeros of the code is obtained as follows. let $\cZ_1=\{1,\ldots,\delta-1\}$. Using \eqref{eq:def-set-2} and setting $\cD_2=\cZ_1$, we have 
    \begin{equation}\label{eq:shifts}
    \cZ_2=\bigcup_{l=0}^{\nu-1}(\cZ_1+l(r+\delta-1)).
   \end{equation}
Further, let $\cL_{3}=\bigcup_{l=0}^{n-1}(\cZ_2+l(j+1))$ and let $\cD_3=\{1,\ldots,\delta_3-1\},$ 
where 
\begin{align}
\delta_3=(n-k)(j+1)+\delta- \left\lceil \frac{k(j+1)}{r}\right\rceil (\delta-1).\label{eq:db}
\end{align}
 Finally, put $\cZ=\cL_3\cup\cD_3$ and let $\cB$ be the cyclic code with generator polynomial $g_{\cB}(x)=\prod_{t\in \cZ}(x-\alpha^t)$. 
Note for future use that the complement of the set $\cD_3$ in the set of exponents of $\alpha$ has cardinality
   \begin{equation}\label{eq:LD}
   |\bar \cD_3|=n(j+1)-(\delta_3-1)=k(j+1)+\Big(\Big\lceil \frac{k(j+1)}{r}\Big\rceil-1\Big) (\delta-1).
   \end{equation}

In the next theorem we give an explicit representation of the code $\cB$ in quasicyclic form, and we also specify its locality properties.
\begin{theorem}
	\label{thm:quasi}
$(a)$	The code ${\cB}$ is an $(n(j+1),k(j+1))$ optimal LRC code with $(r,\delta)$ locality, and the punctured codes
	 $$
	{\cB}_l=\{(c_l,c_{l+n},\ldots,c_{l+nj}) \mid (c_0,\ldots,c_{n(j+1)-1})\in{\cB}\},\quad l=0,\ldots,n-1,
	$$ 
are LRC codes of length $j+1$ with $(r,\delta)$ locality. 
	
$(b)$ Furthermore, if $k| n,$ then the code $\cB$ is equivalent to a code with generator matrix $G$ given by
    $$G=(G_{il}), i=0,\dots,k-1, l=0,\dots n-1,$$ 
where every $G_{il}$ is a $(j+1)\times(j+1)$ circulant matrix. For every $i=0,1,\dots,k-1$ the matrices $G_{il}$ satisfy
  $$
  G_{il}=\begin{cases}I_{j+1} &\text{if } l=in/k\\ 0 &\text{if }l\in\{0,n/k,\dots,n-n/k\}\backslash\{ in/k\}.
  \end{cases}
  $$
\end{theorem}

\begin{proof} $(a)$ We note that the set of zeros of the code $\cB$ is partitioned into segments of length $r+\delta-1,$
i.e., has the structure of the set $\cL$ in \eqref{eq:rts}. In other words, the generator polynomial of $\cB$ 
satisfies 
   $$
	\prod_{t\in\cL}(x-\alpha^t)|g_{\cB}(x), \text{ where }\cL=\bigcup_{s=0}^{n\nu-1}(\cZ_1+s(r+\delta-1)).
	$$
Therefore, Lemma \ref{le:local0} implies that the code $\cB$ has $(r,\delta)$ locality. To compute the dimension of the code $\cB$ we count the 
number of its nonzeros. They are all located in $\bar\cD_3.$ This is a consecutive segment of exponents, and by \eqref{eq:shifts}, within each 
whole subsegment of length $r+\delta-1$ in it there are $r$ nonzeros. Once all such segments are accounted for, there may be an incomplete segment left, which contains $\min(|\bar\cD_3|-\lfloor\frac{|\bar\cD_3|}{r+\delta-1}\rfloor(r+\delta-1),r)$ zeros. 
As easily checked, the total number of nonzeros in either case is $k(j+1)$, which is therefore the dimension of the code $\cB$.
The distance of $\cB$ is at least $\delta_3$, and the bound \eqref{eq:delta} implies that the code $\cB$ has the largest possible distance for
the chosen locality parameters. 

Examining the structure of zeros of the punctured codes $\cB_l,$ we observe that they satisfy the assumptions of Lemma \ref{le:local0}, and
thus the punctured codes ${\cB}_l$ also have $(r,\delta)$ locality. 
Indeed, let $l=0.$ By Lemma~\ref{le:local0}, Eq.~\eqref{eq:C0}, the code $\cB^{(0)}$ given by 
\begin{align*}
	\cB^{(0)} = \{(c_0,c_{n\nu},\ldots,c_{(r+\delta-2)n\nu})\mid (c_0,\ldots,c_{n(j+1)-1})\in\cB \}
\end{align*} 
has dimension at most $r$ and minimum distance at least $\delta$. This implies that the coordinates that are multiples 
of $n\nu$ isolate a repair group of the code $\cB_0$. By shifting this set of coordinates to the right by $n$ positions, we obtain another 
repair group of $\cB_0$, which is disjoint from the first one. After several more shifts we will reach 
the set of coordinates $\{n(\nu-1),n(\nu-1)+n\nu,\ldots,n(\nu-1)+(r+\delta-2)n\nu\}.$ The collection of the sets constructed along the way forms 
a partition of the support of $\cB_0$ into disjoint repair groups. 
The same argument works for every code $\cB_l,1\leq 
l\leq n-1$ whose repair groups are formed by shifting the repair groups of $\cB_0$ to the right by $l$ positions.
This concludes the proof of 
Part (a).
	
Let us prove Part (b). Recalling the discussion in the beginning of Sec.~\ref{sec:clrc}, it is possible to represent
the generator matrix $G$ of the code to have rows of the form
   \begin{align}
		((\alpha^t)^{n(j+1)-1},(\alpha^t)^{n(j+1)-2},\ldots,1), \quad t\in\{0,\ldots,n(j+1)-1\}\setminus\cZ.\label{eq:rows}
	\end{align}   
(note the inverse order of the exponents). 
 Let 
    $$
    \cI=\{0,n/k,\dots,n(j+1)-n/k\}
    $$
be a subset of coordinates. As before, we label the columns of $G$ by the exponents of $\alpha$ from $0$ to $n(j+1)-1$ and consider a square $k(j+1)\times k(j+1)$ submatrix $G_{\cI}$ formed of the columns with indices in $\cI$. We claim that $G_{\cI}$ is invertible. Indeed, the rows of $G_\cI$ have the form
   \begin{align*}
		\alpha^{-t}((\alpha^{tn/k})^{k(j+1)},(\alpha^{tn/k})^{k(j+1)-1},\ldots,\alpha^{tn/k}), \quad t\in\{0,\ldots,n(j+1)-1\}\setminus\cZ,
	\end{align*}
and thus it forms a Vandermonde matrix generated by the set $(\alpha^{tn/k})$ for all $t$ outside the set of zeros $\cZ.$ We will assume that the submatrix $G_{\cI}=I_{k(j+1)}$ (the identity matrix) and continue to use the notation $G$ for the resulting generator matrix of the code.

Let $g_{i,0},g_{i,1},\dots,g_{i,n(j+1)-1}$ be the $i$th row of $G$, where $i\in\{0,\ldots,k-1\}$. 
With an outlook of constructing convolutional codes later in this section, define the polynomials 
   \begin{gather}
   g_i(D)=\sum_{s=0}^{n(j+1)-1}g_{i,s}D^s \nonumber\\
   g_{i,l}(D)=\sum_{s=0}^jg_{i,ns+i}D^i, \quad l=0,1,\dots,n-1. \label{eq:gil}
   \end{gather}
Then we have
	\begin{align*}
		{g}_i(D)=\sum_{l=0}^{n-1} D^l\sum_{s=0}^jg_{i,ns+l}D^{ns}=\sum_{l=0}^{n-1} D^l g_{i,l}(D^n),
	\end{align*}
Since $G_\cI$ is the identity matrix, we have $g_{i,in/k}(D)=1$ and $g_{i,l}(D)=0$ for $l\in\{0,n/k,\ldots,n-n/k\}\backslash\{in/k\}$.

To write the generator matrix in the circulant form given in the statement, we need to form the matrices $G_{il}.$ This is
accomplished by writing the coefficients of $g_{i,j}(D)$ as the first row of $G_{il}$ and filling the rest of this matrix by consecutive 
cyclic shifts to the right. This yields the following $(j+1)\times n(j+1)$ matrix
	\begin{align}
	\begin{pmatrix}
		G_{i,0} & G_{i,1} & \cdots & G_{i,n-1}
	\end{pmatrix}.\label{eq:mat}
	\end{align}
Note that each row in this matrix is a codeword of the code (equivalent to) $\cB.$ Finally, the matrix
	\begin{equation}\label{eq:Gm}
		G=\begin{pmatrix}
		G_{0,0} & G_{0,1} & \cdots & G_{0,n-1}\\
		G_{1,0} & G_{1,1} & \cdots & G_{1,n-1}\\
		\vdots & \vdots & \ddots & \vdots\\
		G_{k-1,0} & G_{k-1,1} & \cdots & G_{k-1,n-1}
		\end{pmatrix}.
	\end{equation}
formed of the rows \eqref{eq:mat} for $i=0,\dots,k-1$ generates a code equivalent to the code $\cB.$
\end{proof}

This theorem gives an explicit representation of $\cB$ as a quasicyclic code, and we can use this representation
to construct a convolutional code following the method in Sec.~\ref{sec:cc-qc}. Let $G(D)=(g_{i,l}(D))$ be a $k\times n$ generator matrix 
where $g_{i,l}(D)$ is defined in \eqref{eq:gil}. Since $\deg(g_{i,l}(D))\le j,$ we conclude that the memory of the generator 
matrix $G(D)$ is $M\leq j$. 
	Having \eqref{eq:qc} in mind, define an $(n,k)$ tailbiting convolutional code over $\cC=\cC_{[0,j]}\in\ff_q[D]$ as a set of sequences
\begin{align*}
	\cC=\left\{c(D)\mid c(D)=u(D)G(D);u(D)=\sum_{s=-M}^{j}u_sD^s; u_{-s}=u_{j+1-s}, s=1,\ldots,M \right\}.
\end{align*}
Here $j, k-1\le j\le n-1$ is any integer such that $(r+\delta-1)|(j+1)$ and $k|n.$

The next theorem states the main properties our construction.
\begin{theorem}
	The code $\cC$ has $(r,\delta)$ row locality. When viewed as a block code, the minimum distance of $\cC$ attains the bound \eqref{eq:delta}.                      
\end{theorem}
\begin{proof}	
For $l=0,\ldots,n-1$, we have
	\begin{align*}
		c^{(l)}(D) = \sum_{i=0}^{k-1}u^{(i)}(D)g_{i,l}(D).
	\end{align*}
	Furthermore, since $u_{-s}=u_{j+1-s}$ for $s=1,\ldots,M$, we have the following relation
	\begin{align*}
		c^{(l)}(D) = \sum_{i=0}^{k-1}u^{(i)}(D)g_{i,l}(D) \mod (D^{j+1}-1).
	\end{align*}
	In other words, we have
	\begin{align*}
		\begin{pmatrix}
		c^{(0)}  c^{(1)}  \dots  c^{(n-1)}
		\end{pmatrix}=
		\begin{pmatrix}
		u^{(0)}  u^{(1)}  \dots  u^{(k-1)}
		\end{pmatrix} G, 
	\end{align*}
where the matrix $G$ is defined in \eqref{eq:Gm}. This implies that the codes $\cC^{(l)}$ are exactly the codes $\cB_l$ defined in
in Theorem~\ref{thm:quasi}(a), viz., $\cC^{(l)}={\cB}_l$ for $l=0,\ldots,n-1.$ Since the code ${\cB}_l$ has $(r,\delta)$ locality for $l=0,\ldots,n-1$, we conclude that 
the code $\cC$ has $(r,\delta)$ row locality. Concluding, we have established that the convolutional code
$\cC$ has $(r,\delta)$ row locality.

As a block code, $\cC$ is equivalent to the code $\cB,$ which proves the last claim of the theorem.
\end{proof}

The large minimum distance of the code $\cC$ is related to the performance of the (hard decision) Viterbi decoding of the code $\cC$, and is therefore of interest in applications.

As remarked earlier, the constructed codes stop short of attaining the bound \eqref{eq:sb-lrcc}, and thus cannot be claimed to be optimal. 
Of course, as observed after Def.~\ref{def:djc}, the $j$th column distance of the code $\cC$ is at least the minimum distance of the code $\cB$, 
given by \eqref{eq:db}, but a more precise estimate remains an open question. 
Nevertheless, we believe that extension of the basic construction of LRC codes to the case of convolutional codes carries potential for future 
research into their structure. In particular, the algebraic machinery of quasicyclic codes of \cite{Ling01,Lally01} could lead to new constructions, 
and it may also be possible to further extend these studies to codes over ring alphabets \cite{Ling2003}.

\section{Bi-cyclic codes with availability}\label{sec:avl}

The H-LRC codes constructed in Sec.~\ref{sec:h} rely on $h$ embedded recovering sets for each code coordinate, which are not disjoint. In this section we consider LRC codes with $t$ \emph{disjoint} recovering sets for each code coordinate, i.e., LRC codes with availability $t$ (here we do not pursue a hierarchy of locality). This arises when the data is simultaneously requested by a large number of users, 
which suggests that the erased coordinate afford recovery from several nonoverlaping recovering sets in order to increase data availability. LRC codes with this property are defined as follows \cite{wang2014repair,rawat2016locality}.

\begin{definition}[\sc LRC codes with availability]
	\label{def:lrc-a}
	Let $t\geq 1$ and $(r_1,\delta_1),\ldots,(r_t,\delta_t)$ be integers. A code $\cC\subset \mathbb{F}_q^n$ is said to have availability $t$ with locality $(r_1,\delta_1),\ldots,(r_t,\delta_t)$ if for every coordinate $j,1\leq j\leq n$ of the code $\cC$ there are $t$ punctured codes $\cC^{(i)},1\leq i\leq t$ such that $j\in\supp(\cC^{(i)}),i=1,\dots,t$ and \\
	\indent (a) $\dim(\cC^{(i)})\le r_i$,\\
	\indent (b) $d(\cC^{(i)})\ge\delta_i$,\\
	\indent (c) $\bigcap_{i=1}^t \supp(\cC^{(i)})=\{j\}$.
\end{definition}

Below we limit ourselves to the case $t=2.$
We say that two partitions $\cP_1,\cP_2$ of the coordinates of the code are {\em orthogonal} (or transversal) if $|P_1\cap P_2|\le 1$
for any $P_1\in \cP_1,P_2\in \cP_2$ and every coordinate is contained in a pair of subsets $X\in\cP_1, Y\in\cP_2.$
Orthogonal partitions enable multiple disjoint recovering sets and were used in \cite{tb14} to construct codes with availability. 
A simple observation made in \cite{tb14} is that product codes naturally yield orthogonal partitions, and it is possible to use
products of one-dimensional cyclic codes to support this structure. A drawback of this approach is that product codes result in rather poor 
parameters of LRC codes with availability, in particular the rate of the resulting codes is low (although the alphabet is small compared to the 
code length \cite{tb14,huang2017multierasure}). It is well known that the rate of product codes can be increased with no loss to the distance by passing to generalized concatenations of codes \cite{BZ1974}. In this section we use a particular case of this construction given in \cite{saints1993hyperbolic} and sometimes called {\em hyperbolic codes}. The resulting LRC codes with availability have the same distance guarantee as simple product codes while having a much higher rate. As above, our starting point is the general method of Theorem \ref{thm:cyclic}, and we proceed similarly to Eq.~\eqref{eq:def-set-2}. 

We start with a finite field $\ff_q$ and assume that the code length $n$ divides $q-1.$ We further choose the size of the repair groups
to be $r_1,r_2$ and suppose that $0<r_1\leq r_2$ and $(r_1+1)| n$ and $(r_2+1)| n$. Further, let $\nu_1=n/(r_1+1)$ and $\nu_2=n/(r_2+1)$. Let $\alpha\in\ff_q$ be a primitive $n$-th root of unity. 
To simplify the expressions below, we will construct codes with $\delta_1=\delta_2=2.$ To construct the defining set $\cZ$ of our code, let
\begin{align}
\cL_1=\emptyset,\quad \cD_1=\{(0,0)\},\quad \cZ_1=\cL_1\cup\cD_1.
\end{align}
Let us fix the designed distance of the code $\cC$ to be $\delta\geq 2.$ Define
\begin{align*}
&\cL_{2,1}=\bigcup_{l_1=0}^{\nu_1-1}\bigcup_{j=0}^{n-1}(\cZ_1+(l_1(r_1+1),j)),\;\\
&\cL_{2,2}=\bigcup_{i=0}^{n-1}\bigcup_{l_2=0}^{\nu_2-1}(\cZ_1+(i,l_2(r_2+1))),\;
\cL_2=\cL_{2,1}\cup\cL_{2,2},\\
&\cD_2=\{(i,j) \mid (i+1)(j+1) < \delta \},\;
\cZ_2=\cL_2\cup\cD_2.
\end{align*}
Note that zeros are now indexed by pairs of exponents, and pairs are
added element-wise. Finally, put $\cZ=\cZ_2$. 

Consider a two-dimensional cyclic code $\cC=\langle g(x,y)\rangle$ of length $n^2$, where
\begin{align}
g(x,y) = \prod_{(i,j)\in\cZ}(x-\alpha^i)(y-\alpha^j).
\end{align}

\begin{lemma}
	The code $\cC$ has two disjoint recovering sets of size $r_1$ and $r_2$ for every coordinate.
\end{lemma}

\begin{proof}
	The proof relies on Lemma~\ref{le:local0}. For $c\in\cC$, let us write $c=(c_{i,j})$ where $0\leq i\leq n-1,0\leq j\leq n-1$. Fix $j$ and let $\cC^{(1)}=\{ (c_{0,j}, c_{\nu_1,j},\ldots,c_{r_1\nu_1,j}) \mid c\in \cC \}$ and let $\cL_{2,1}^{j}=\bigcup_{l_1=0}^{\nu_1-1}(\cZ_1+(l_1(r_1+1),j))$. The generator polynomial for the punctured code $\{ (c_{0,j}, c_{1,j},\ldots,c_{n-1,j}) \mid c\in \cC \}$ is given by $y^j\prod_{(i,j)\in\cZ}(x-\alpha^i)$, which is divisible by $\prod_{(i,j)\in\cL_{2,1}^{j}}(x-\alpha^i)$. Then, by arguments similar to Lemma~\ref{le:local0} we
conclude that the code $\cC^{(1)}$ has dimension at most $r_1$ and distance at least $2$. Since the code is cyclic in both dimensions, we 
also claim that every of symbol of the code $\cC$ has a recovering set of size at most $r_1$ for one erasure.
	
Repeating the above arguments for a fixed index $i$, we isolate another recovering of size $r_2$ for every coordinate. Furthermore, the two recovering sets are disjoint by construction.
\end{proof}

To estimate the distance, recall the following result about bicyclic codes.

\begin{lemma}[\sc Hyperbolic bound \cite{saints1993hyperbolic}]
	Suppose that the defining set of zeros of an $n\times n$ bicyclic code contains a subset given by $\cZ=\{(i,j): (i+1)(j+1) < d \}.$ Then the 
distance of the code is at least $d$.
\end{lemma}

Thus, the distance of the code $\cC$ constructed above is at least $\delta,$ and its dimension $\dim(\cC)=n^2-|\cZ|.$  
Let us estimate the dimension from below. We have
\begin{align*}
|\cD_2| &= \sum_{j=0}^{\delta-2}\left\lfloor \frac{\delta}{j+1} \right\rfloor
\leq \delta\sum_{j=1}^{\delta-1}\frac{1}{j}\\
&\leq \delta\left(1+\int_{1}^{\delta-1}\ln (x-1) dx\right)\\
&\le \delta(1+\ln{(\delta-1)}).
\end{align*}
Therefore, we have
\begin{align*}
n^2-|\cZ| &\geq n^2-|\cL_2|-|\cD_2|\\
&= n^2-(r_1+r_2+1)\nu_1\nu_2-\delta(1+\ln{(\delta-1)})\\
&= \frac{n^2r_1r_2}{(r_1+1)(r_2+1)}-\delta(1+\ln{(\delta-1)}).
\end{align*}

To compare this estimate of the dimension with product codes, let $\cC'$ be a direct product of two cyclic LRC codes of 
length $n$ with locality $r$ and distance $\sqrt{\delta}$. The defining set of this code are given by $\cZ' = \cL'\cup\cD'$, where
\begin{align*}
	\cL' &= \{(i,j) \mid i,j = 1 \bmod (r+1)\},\\
	\cD' &= \{(i,j) \mid i,j = 1,\ldots,\sqrt{\delta}-1 \}.
\end{align*}
Choosing $r_1=r_2=r$ in our construction, we have $\cL_2+(1,1) = \cL'$. It is also easy to see that $|\cZ|\leq |\cZ'|$ and thus, $\dim(C)\geq \dim(\cC')$.

\begin{example}
	Let $r_1=2$, $r_2=6$, and $\delta=9$. Our construction gives rise to an LRC code with availability two, of length $n^2=441$, dimension $k=246$, and distance $d\geq9$, over a finite field $\mathbb{F}_q$ such that $21|(q-1)$ (for example, we can take $q=64$).
To compare these codes with the product construction, let us choose the column codes and row codes of length 21
and distance $3$, and let us take the maximum dimension of codes as given by the bound \eqref{eq:sb-1} for locality 2 and 6. This gives 
$k_1=13,k_2=17$ and the overall dimension $k=221,$ lower than above.
\end{example}

An extension of the construction in this section to the case of $t\geq 2$ can be easily obtained via $t$-dimensional cyclic codes and the general hyperbolic bound. 
Furthermore, using procedures similar to those used in \eqref{eq:def-set-1} and \eqref{eq:def-set-2}, our construction can be 
generalized to $h\ge 2$ levels of hierarchy such that the local codes in each of the $h$ levels have availability $t\ge 2.$ 
\appendices
\section{Proof of Proposition~\ref{cl:card}}\label{app:card}
Let us first prove a technical claim.
\begin{claim}
	\label{cl:d-mod}
	For $1\leq i\leq h+1$, we have $\delta_1 = (\delta_i-b^{(i-1)}_0) \bmod n_1$.
\end{claim}
\begin{proof}
We prove the claim by induction. Clearly, for $i=1$ we have $\delta_1 = (\delta_1-b^{(0)}_0) \bmod n_1$. 
Next suppose that for a given $i, 1\leq i<h+1$ we have $\delta_1 = (\delta_{i}-b^{(i-1)}_0) \bmod n_1$.
Reducing \eqref{eq:d} modulo $n_1,$ we obtain the equality $(\delta_i-b^{(i-1)}_0) = (\delta_{i+1}-b^{(i)}_0)\bmod n_1$.
Thus also $\delta_1 = (\delta_{i+1}-b^{(i)}_0)\bmod n_1$, and this completes the
proof. \end{proof}

Now we are ready to prove Proposition~\ref{cl:card}.  
Clearly, for $i=1$ we have $|\cZ_1|=\delta_1-1=n_1-r_1$.
Suppose for $1\leq i'\leq i$ we have $|\cZ_{i'}|=n_{i'}-r_{i'}$, where $1\leq i<h+1$.
Let us establish the induction step. 

By the definition of $\cZ_{i+1}$ and \eqref{eq:d}, we have $|\cZ_{i+1}| \geq (\nu_i-a_i)n_i$. 
Consider the set 
   $$
\cA_{i+1}=(\nu_i-a_i)n_i+\{\delta_i-b^{(i-1)}_0-\delta_1+1,\delta_i-b^{(i-1)}_0-\delta_1+2,\ldots,n_i\}
   $$    
formed by adding $(\nu_i-a_i)n_i$ to every element of the subset above, and let $\cB_{i+1}=\cA_{i+1}\cap\cD_{i+1}$. 
Since $\cup_{l=0}^{\nu_i-a_i-1}(\cZ_i+ln_i)\subseteq \cD_{i+1}$, we have
\begin{align}
|\cZ_{i+1}| &= (\nu_i-a_i)n_i+|\cB_{i+1}\setminus (\cZ_i+(\nu_i-a_i)n_i)|+\Big|\bigcup_{l=\nu_i-a_i}^{\nu_i-1}(\cZ_i+ln_i)\Big|.
\end{align}
Next, we would like to determine the cardinality of the set $\cB_{i+1}\setminus (\cZ_i+(\nu_i-a_i)n_i)$. By Claim~\ref{cl:d-mod}, we have $\delta_1 = (\delta_i-b^{(i-1)}_0) \bmod n_1$. Therefore, $|\cA_{i+1}|$ is divisible by $n_1$. Note that $\cB_{i+1}\subseteq\cA_{i+1}$. Moreover, among the first $u^{(i)}_{i-1}n_{i-1}$ elements of $\cA_{i+1}$ there are $u^{(i)}_{i-1}r_{i-1}-b^{(i-1)}_0$ elements that are in $\cB_{i+1}$ but not in $\cZ_i+(\nu_i-a_i)n_i$. For the next $u^{(i)}_{i-2}n_{i-2}$ elements in $\cA_{i+1}$ there are $u^{(i)}_{i-2}r_{i-2}$ elements that are in $\cB_{i+1}$ but not in $\cZ_i+(\nu_i-a_i)n_i$, and so forth.
Hence, we have
\begin{align*}
|\cB_{i+1}\setminus (\cZ_i+(\nu_i-a_i)n_i)|
& = \sum_{j=1}^{i-1}u^{(i)}_j r_j + u^{(i)}_0 r_1+b^{(i)}_0 -b^{(i-1)}_0\\
& = b_i.
\end{align*} 
It follows that
\begin{align}
|\cZ_{i+1}| &= (\nu_i-a_i)n_i+b_i+\Big|\bigcup_{l=\nu_i-a_i}^{\nu_i-1}(\cZ_i+ln_i)\Big|\nonumber\\
&=(\nu_i-a_i)n_i+b_i+a_i|\cZ_i|\nonumber\\
&=n_{i+1}-a_i n_i+b_i+a_i(n_i-r_i)\label{eq:hypo}\\
&=n_{i+1}-(a_i r_i-b_i)\nonumber\\
&=n_{i+1}-r_{i+1}\nonumber,
\end{align}
where \eqref{eq:hypo} uses $|\cZ_i|=n_i-r_i$, which is the induction hypothesis. This completes the proof.

\section{Proof of Lemma~\ref{le:opt}}\label{app:le:opt}
We again argue by induction on $i$. For $i=1$, by \eqref{eq:d} we have
\begin{align}
\delta_{2} &= (\nu_1-a_1)n_1+\delta_1 + u^{(1)}_0 n_1+b^{(1)}_0-b^{(0)}_0\nonumber\\
&= n_2-a_1n_1+\delta_1+ u^{(1)}_0 (n_1-r_1) +b_1\label{eq:1}\\
&= n_2-a_1(r_1+\delta_1-\delta_0)+\delta_1 + u^{(1)}_0 (n_1-r_1) +b_1\nonumber\\
&= n_2-r_2+\delta_1-a_1(\delta_1-\delta_0) + u^{(1)}_0 (\delta_1-\delta_0)\label{eq:2}\\
&= n_2-r_2+\delta_1-(a_1-u^{(1)}_0)(\delta_1-\delta_0)\nonumber\\
&= n_2-r_2+\delta_1-\left\lceil \frac{r_2}{r_1}\right\rceil(\delta_1-\delta_0),\nonumber
\end{align}
where \eqref{eq:1} follows from $n_2=n_1\nu_1$ and \eqref{eq:b0}, in \eqref{eq:2} we used $a_1r_1=r_2+b_1$, and the last equality follows from $a_1=\lceil r_2/r_1 \rceil$ and $u^{(1)}_0=\lfloor (b^{(1)}_1+b^{(0)}_0)/r_1\rfloor=0$.

For the induction step, let us fix $i, 1\le i<h$ and suppose that 
     \begin{align}\label{eq:ih}
\delta_{i+1} = n_{i+1}-r_{i+1}+\delta_{i}-\sum_{l=1}^{i}\left\lceil \frac{r_{i+1}}{r_l} \right\rceil (\delta_l-\delta_{l-1}),
     \end{align}
provided that conditions \eqref{eq:opt} are satisfied. Observe that
\begin{align}
\delta_{i+2} &= (\nu_{i+1}-a_{i+1})n_{i+1}+\delta_{i+1}+\sum_{j=1}^{i}u^{(i+1)}_j n_j + u^{(i+1)}_0 n_1+b^{(i+1)}_0-b^{(i)}_0\nonumber\\
&= n_{i+2}-a_{i+1}n_{i+1}+\delta_{i+1}+\sum_{j=1}^{i}u^{(i+1)}_j (n_j-r_j) + u^{(i+1)}_0 (n_1-r_1) +b_{i+1}\label{eq:step1}
\end{align}
Substituting $n_{i+1}$ from \eqref{eq:ih}, we obtain
\begin{align}
a_{i+1}n_{i+1} & = a_{i+1}\Big( r_{i+1} + \delta_{i+1} - \delta_{i}+\sum_{l=1}^{i}\Big\lceil \frac{r_{i+1}}{r_l} \Big\rceil (\delta_l-\delta_{l-1}) \Big)\label{eq:step2}.
\end{align}
In addition, also by the induction hypothesis, we have
\begin{align*}
u^{(i+1)}_0(n_1-r_1)&= u^{(i+1)}_0(\delta_1-\delta_0),\\
u^{(i+1)}_j(n_j-r_j)&= u^{(i+1)}_j (\delta_j-\delta_{j-1}) + \sum_{l=1}^{j-1} u^{(i+1)}_j\Big\lceil \frac{r_j}{r_l} \Big\rceil(\delta_l-\delta_{l-1}),\quad 1\leq j \leq i.
\end{align*}
Therefore, 
\begin{align}
\sum_{j=1}^{i}u^{(i+1)}_j & (n_j-r_j) + u^{(i+1)}_0 (n_1-r_1)\nonumber\\
&=\sum_{l=1}^{i} \Big( u^{(i+1)}_l + \sum_{j=l+1}^{i}u^{(i+1)}_j\Big\lceil \frac{r_j}{r_l} \Big\rceil \Big)(\delta_l-\delta_{l-1})+ u^{(i+1)}_0(\delta_1-\delta_0). \label{eq:step3}
\end{align}
Substituting \eqref{eq:step2} and \eqref{eq:step3} into \eqref{eq:step1}, we obtain
\begin{align*}
\delta_{i+2} = n_{i+2}- & r_{i+2}+\delta_{i+1}-a_{i+1}(\delta_{i+1}-\delta_{i})\\
&+\sum_{l=2}^{i}\Big(  u^{(i+1)}_l + \sum_{j=l+1}^{i}u^{(i+1)}_j\Big\lceil \frac{r_j}{r_l} \Big\rceil - a_{i+1} \Big\lceil \frac{r_{i+1}}{r_l} \Big\rceil \Big) (\delta_l-\delta_{l-1})\\
&+\Big(  u^{(i+1)}_0 + u^{(i+1)}_1 + \sum_{j=2}^{i}u^{(i+1)}_j\Big\lceil \frac{r_j}{r_1} \Big\rceil - a_{i+1} \Big\lceil \frac{r_{i+1}}{r_1} \Big\rceil \Big) (\delta_1-\delta_{0}).
\end{align*}
Thus, if the corresponding conditions of \eqref{eq:opt} are satisfied, then we have
\begin{align*}
\delta_{i+2} &= n_{i+2}-  r_{i+2}+\delta_{i+1}-\sum_{l=1}^{i+1}\Big\lceil \frac{r_{i+2}}{r_l}\Big\rceil (\delta_{i+1}-\delta_{i}).
\end{align*}
This completes the induction step.

\section{Proof of Proposition~\ref{prop:lrcc-max}}\label{app:proof-of-lrcc-max}

\noindent {\em Part (a)}: 
The generator matrix of the truncated code $\cC_{[0,j]}$ is given in \eqref{eq:Gj}, where $\rank(G_0)=k$.

Since $G_0$ has full rank, by Definition~\ref{def:djc}, the $j$-th column distance of $\cC$ is equal to
\begin{align*}
	d_j^c=\min\{\mathrm{wt}(u_{[0,j]}G_j^c) \mid u_0\neq 0 \},
\end{align*} 
where $u_{[0,j]}=(u_0,\ldots,u_j)\in\ff_q^{k(j+1)}$ is an input sequence truncated at the $j$-th time instant. 
Obviously,  
\begin{align}
	d_j^c\leq \min\{\mathrm{wt}(u_{[0,j]}G_j^c) \mid u_0\neq 0, u_1=u_2=\ldots=u_j=0 \}. \label{eq:djc-u}
\end{align}
Consider the $k\times n(j+1)$ submatrix $[G_0, G_1,\dots,G_j]$ that forms the first $k$ rows of $G_j^c$. By elementary row operations 
it is possible possible to make some, say first, $k$ columns of $G_1$ be all-zero, and the same is true for some $k$ columns of the matrices $G_i,i=2,\dots,j.$ (Note that to accomplish this, we use all the rows of $G_j^c$ and not just the rows of the submatrix $[G_0, G_1,\dots,G_j]$.)
As a result, the vector $b(u_0):=u_0[G_0, G_1,\dots,G_j]$ will contain at least $kj$ zero coordinates for any choice of $u_0\in \ff_q^k$, and so the
effective length of the set of vectors $\{b(u_0)\}$ is $n(j+1)-kj$. The set 
of vectors $\{b(u_0) \mid u_0\in\ff_q^k\}$ is a subset of the code $\cC_{[0,j]}$ and it forms a linear code $\cK$ of effective length $n(j+1)-kj$ with $(r,\delta)$ locality if the original convolutional code $\cC$ has (column or row) $(r,\delta)$ locality. The distance of the
code $\cK$ gives an upper bound on $d_j^c,$ and is itself bounded above as in \eqref{eq:delta}. Substituting the parameters of the code $\cK,$ we obtain the bound \eqref{eq:sb-lrcc} for the distance $d_j^c.$ Part (a) is proved.

{\em Remark:} Without the locality assumption the above argument proves the Singleton bound \eqref{eq:sbc} for the column distance of the code $\cC.$

\vspace*{.1in}\noindent {\em Part (b)}: Recall the following observation (e.g., \cite{gluesing2006strongly}):

\begin{proposition}\label{prop:djc}
	Let $j\geq 0$. Let $H_j^c$ be the parity check matrix for the truncated convolutional code $\cC_{[0,j]}$. Then the following properties are equivalent:
	\begin{enumerate}
		\item $d_j^c=d$;
		\item none of the first $n$ columns of $H_j^c$ is contained in the span of any other $d-2$ columns and one of the first $n$ columns of $H_j^c$ is in the span of some other $d-1$ columns of $H_j^c$.
	\end{enumerate}
\end{proposition}

Now suppose that $d_j^c$ attains \eqref{eq:sb-lrcc} with equality while $d_{j-1}^c<(n-k)j+\delta-\left\lceil\frac{k}{r}\right\rceil(\delta-1)$. By Proposition~\ref{prop:djc}, there exists a column among the first $n$ columns of $H_{j-1}^c$ such that it is in the span of some other $d_{j-1}^c-1$ columns of $H_{j-1}^c$. Note that $H_{j-1}^c$ is a submatrix of $H_j^c$. Specifically, we have
\begin{align*}
	H_j^c=\begin{pmatrix}
	H_{j-1}^c & \begin{matrix} 0 \\ \vdots \\ 0 \end{matrix}\\
	\begin{matrix} H_j & H_{j-1} & \cdots & H_1 \end{matrix} & H_0
	\end{pmatrix},
\end{align*} where $H_i,1\leq i \leq j$ are $(n-k)\times n$ matrices and $\rank(H_0)=n-k.$
Because of this, there exists a column among the first $n$ columns of $H_j^c$ such that it is in the span of some other $d_{j-1}^c-1+n-k<d_j^c-1$ columns of $H_j^c$, which by Proposition~\ref{prop:djc} contradicts our assumption about $d_j^c$.
Hence, it follows that the optimality of the $j$-th column distance implies the optimality of the $i$-th column distance for all $i\leq j$ for convolutional codes with (column or row) locality.

	\bibliographystyle{IEEEtranS}
	\bibliography{./LRCMultiSets}

\begin{thebibliography}{10}
\providecommand{\url}[1]{#1}
\csname url@samestyle\endcsname
\providecommand{\newblock}{\relax}
\providecommand{\bibinfo}[2]{#2}
\providecommand{\BIBentrySTDinterwordspacing}{\spaceskip=0pt\relax}
\providecommand{\BIBentryALTinterwordstretchfactor}{4}
\providecommand{\BIBentryALTinterwordspacing}{\spaceskip=\fontdimen2\font plus
\BIBentryALTinterwordstretchfactor\fontdimen3\font minus
  \fontdimen4\font\relax}
\providecommand{\BIBforeignlanguage}[2]{{%
\expandafter\ifx\csname l@#1\endcsname\relax
\typeout{** WARNING: IEEEtranS.bst: No hyphenation pattern has been}%
\typeout{** loaded for the language `#1'. Using the pattern for}%
\typeout{** the default language instead.}%
\else
\language=\csname l@#1\endcsname
\fi
#2}}
\providecommand{\BIBdecl}{\relax}
\BIBdecl

\bibitem{bbv19}
S.~Ballentine, A.~Barg, and S.~Vl{\u{a}}du{\c{t}}, ``Codes with hierarchical
  locality from covering maps of curves,'' \emph{IEEE Trans. Inf. Theory},
  vol.~65, no.~10, pp. 6056--6071, 2019.

\bibitem{B+17}
A.~Barg, K.~Haymaker, E.~Howe, G.~Matthews, and A.~V{\'a}rilly-Alvarado,
  ``Locally recoverable codes from algebraic curves and surfaces,'' in
  \emph{Algebraic Geometry for Coding Theory and Cryptography}, E.~Howe,
  K.~Lauter, and J.~Walker, Eds.\hskip 1em plus 0.5em minus 0.4em\relax
  Springer, 2017, pp. 95--126.

\bibitem{BTV17}
A.~Barg, I.~Tamo, and S.~Vl{\u{a}}du{\c{t}}, ``Locally recoverable codes on
  algebraic curves,'' \emph{IEEE Trans. Inform. Theory}, vol.~63, no.~8, pp.
  4928--4939, 2017.

\bibitem{beemer18}
A.~{Beemer}, R.~{Coatney}, V.~{Guruswami}, H.~H. {Lopez}, and F.~{Pinero},
  ``Explicit optimal-length locally repairable codes of distance 5,'' in
  \emph{2018 56th Annual Allerton Conference on Communication, Control, and
  Computing (Allerton)}, 2018, pp. 800--804.

\bibitem{BZ1974}
E.~Blokh and V.~Zyablov, ``Coding of generalized cascade codes,'' \emph{Probl.
  Inform. Trans.}, vol.~10, no.~3, pp. 218--222, 1974.

\bibitem{CaiMiaoSchwartzTang18}
H.~{Cai}, Y.~{Miao}, M.~{Schwartz}, and X.~{Tang}, ``On optimal locally
  repairable codes with super-linear length,'' in \emph{2019 IEEE International
  Symposium on Information Theory (ISIT)}, 2019, pp. 2818--2822.

\bibitem{chen17}
B.~Chen, S.-T. Xia, J.~Hao, and F.-W. Fu, ``Constructions of optimal
  $(r,\delta)$ locally repairable codes,'' \emph{IEEE Trans. Inf. Theory},
  vol.~64, no.~4, pp. 2499--2511, 2018.

\bibitem{Datta2014}
A.~{Datta}, ``Locally repairable rapid{RAID} systematic codes--one simple
  convoluted way to get it all,'' in \emph{2014 IEEE Information Theory
  Workshop (ITW 2014)}, 2014, pp. 60--64.

\bibitem{esmaeili1998link}
M.~Esmaeili, T.~A. Gulliver, N.~P. Secord, and S.~A. Mahmoud, ``A link between
  quasi-cyclic codes and convolutional codes,'' \emph{IEEE Trans. Inform.
  Theory}, vol.~44, no.~1, pp. 431--435, 1998.

\bibitem{freij2016locally}
R.~Freij-Hollanti, T.~Westerb{\"a}ck, and C.~Hollanti, ``Locally repairable
  codes with availability and hierarchy: {M}atroid theory via examples,'' in
  \emph{24th International Z{\"u}rich Seminar on Communications (IZS), Zurich,
  Switzerland, March 2-4, 2016}.\hskip 1em plus 0.5em minus 0.4em\relax
  ETH-Z{\"u}rich, 2016, pp. 45--49, https://doi.org/10.3929/ethz-a-010645448.

\bibitem{gluesing2006strongly}
H.~Gluesing-Luerssen, J.~Rosenthal, and R.~Smarandache, ``Strongly-{MDS}
  convolutional codes,'' \emph{IEEE Trans. Inf. Theory}, vol.~52, no.~2, pp.
  584--598, 2006.

\bibitem{gopalan2012locality}
P.~Gopalan, C.~Huang, H.~Simitci, and S.~Yekhanin, ``On the locality of
  codeword symbols,'' \emph{IEEE Trans. Inf. Theory}, vol.~58, no.~11, pp.
  6925--6934, 2012.

\bibitem{grezet2019complete}
M.~{Grezet} and C.~{Hollanti}, ``The complete hierarchical locality of the
  punctured simplex code,'' in \emph{2019 IEEE International Symposium on
  Information Theory (ISIT)}, 2019, pp. 2833--2837.

\bibitem{GXY19}
V.~Guruswami, C.~Xing, and C.~Yuan, ``How long can optimal locally repairable
  codes be?'' \emph{IEEE Trans. Inf. Theory}, vol.~6, no.~6, pp. 3662--3670,
  2019.

\bibitem{holzbaur2018cyclic}
L.~Holzbaur, R.~Freij-Hollanti, and A.~Wachter-Zeh, ``Cyclic codes with
  locality and availability,'' 2018, arXiv preprint arXiv:1812.06897.

\bibitem{huang2017multierasure}
P.~Huang, E.~Yaakobi, and P.~H. Siegel, ``Multi-erasure locally recoverable
  codes over small fields: {A} tensor product approach,'' \emph{IEEE Trans.
  Inf. Theory}, 2020, {DOI} 10.1109/TIT.2019.2962012.

\bibitem{IKZ2018}
F.~{Ivanov}, A.~{Kreshchuk}, and V.~{Zyablov}, ``On the local erasure
  correction capacity of convolutional codes,'' in \emph{2018 International
  Symposium on Information Theory and Its Applications (ISITA), Singapore},
  2018, pp. 296--300.

\bibitem{JZ15}
R.~Johannesson and K.~S. Zigangirov, \emph{Fundamentals of Convolutional
  Coding}, 2nd~ed.\hskip 1em plus 0.5em minus 0.4em\relax Hoboken, NJ: J. Wiley
  \& Sons, Inc., 2015.

\bibitem{Justesen1990}
J.~{Justesen}, E.~{Paaske}, and M.~{Ballan}, ``Quasi-cyclic unit memory
  convolutional codes,'' \emph{IEEE Transactions on Information Theory},
  vol.~36, no.~3, pp. 540--547, 1990.

\bibitem{kamath2012codes}
G.~M. Kamath, N.~Prakash, V.~Lalitha, and P.~V. Kumar, ``Codes with local
  regeneration and erasure correction,'' \emph{IEEE Trans. Inform. Theory},
  vol.~60, no.~8, pp. 4637--4660, 2014.

\bibitem{Lally01}
K.~Lally and P.~Fitzpatrick, ``Algebraic structure of quasicyclic codes,''
  \emph{Discrete Applied Mathematics}, vol. 111, pp. 157--175, 2001.

\bibitem{Li17a}
X.~Li, L.~Ma, and C.~Xing, ``Construction of asymptotically good locally
  repairable codes via automorphism groups of function fields,'' \emph{IEEE
  Trans. Inf. Theory}, vol.~65, no.~11, pp. 7087--7094, 2019.

\bibitem{LiMaXing17b}
------, ``Optimal locally repairable codes via elliptic curves,'' \emph{IEEE
  Trans. Inform. Theory}, vol.~65, no.~1, pp. 108--117, 2019.

\bibitem{Ling01}
S.~Ling and P.~Sol{\'e}, ``On the algebraic structure of quasi-cyclic codes,
  {I},'' \emph{IEEE Trans. Inf. Theory}, vol.~47, no.~7, pp. 2751--2760, 2001.

\bibitem{Ling2003}
------, ``On the algebraic structure of quasi-cyclic codes {II}: Chain rings,''
  \emph{Designs, Codes and Cryptography}, vol.~30, pp. 113--130, 2003.

\bibitem{luo2018optimal}
Y.~Luo, C.~Xing, and C.~Yuan, ``Optimal locally repairable codes of distance 3
  and 4 via cyclic codes,'' \emph{IEEE Trans. Inf. Theory}, vol.~65, no.~2, pp.
  1048--1053, 2018.

\bibitem{MN19}
U.~Mart{\'\i}nez-Pe{\~n}as and D.~Napp, ``Locally repairable convolutional
  codes with sliding window repair,'' in \emph{Proc. 2019 IEEE Int. Sympos.
  Inf. Theory (ISIT)}, 2019, pp. 2838--2842, expanded version in arXiv preprint
  arXiv:1901.02073.

\bibitem{rawat2016locality}
A.~S. Rawat, D.~S. Papailiopoulos, A.~G. Dimakis, and S.~Vishwanath, ``Locality
  and availability in distributed storage,'' \emph{IEEE Trans. inf. theory},
  vol.~62, no.~8, pp. 4481--4493, 2016.

\bibitem{saints1993hyperbolic}
K.~Saints and C.~Heegard, ``On hyperbolic cascaded {R}eed-{S}olomon codes,'' in
  \emph{International Symposium on Applied Algebra, Algebraic Algorithms, and
  Error-Correcting Codes}.\hskip 1em plus 0.5em minus 0.4em\relax Springer,
  1993, pp. 291--303.

\bibitem{SAC2015}
B.~{Sasidharan}, G.~K. {Agarwal}, and P.~V. {Kumar}, ``Codes with hierarchical
  locality,'' in \emph{2015 IEEE International Symposium on Information Theory
  (ISIT)}, 2015, pp. 1257--1261, expanded version in arXiv preprint
  arXiv:1501.06683.

\bibitem{SvT79}
G.~Solomon and H.~C.~A. van Tilborg, ``A connection between block and
  convolutional codes,'' \emph{SIAM J. Appl. Math.}, vol.~37, no.~2, pp.
  358--369, 1979.

\bibitem{tb14}
I.~Tamo and A.~Barg, ``A family of optimal locally recoverable codes,''
  \emph{IEEE Trans. Inf. Theory}, vol.~60, no.~8, pp. 4661--4676, 2014.

\bibitem{tamo2016bounds}
I.~Tamo, A.~Barg, and A.~Frolov, ``Bounds on the parameters of locally
  recoverable codes,'' \emph{IEEE Trans. inf. theory}, vol.~62, no.~6, pp.
  3070--3083, 2016.

\bibitem{CyclicLRC16}
I.~Tamo, A.~Barg, S.~Goparaju, and R.~Calderbank, ``Cyclic {LRC} codes, binary
  {LRC} codes, and upper bounds on the distance of cyclic codes,''
  \emph{International Journal of Information and Coding Theory}, vol.~3, no.~4,
  pp. 345--364, 2016.

\bibitem{tamo2016cyclic}
------, ``Cyclic {LRC} codes, binary {LRC} codes, and upper bounds on the
  distance of cyclic codes,'' \emph{International Journal of Information and
  Coding Theory}, vol.~3, no.~4, pp. 345--364, 2016.

\bibitem{Tanner2004}
R.~M. {Tanner}, D.~{Sridhara}, A.~{Sridharan}, T.~E. {Fuja}, and D.~J.
  {Costello}, ``{LDPC} block and convolutional codes based on circulant
  matrices,'' \emph{IEEE Transactions on Information Theory}, vol.~50, no.~12,
  pp. 2966--2984, 2004.

\bibitem{TRV2012}
V.~Tom{\'a}s, J.~Rosenthal, and R.~Smarandache, ``Decoding of convolutional
  codes over the erasure channel,'' \emph{IEEE Trans. Inf. Theory}, vol.~58,
  no.~1, pp. 90--108, 2012.

\bibitem{wang2014repair}
A.~Wang and Z.~Zhang, ``Repair locality with multiple erasure tolerance,''
  \emph{IEEE Trans. Inf. Theory}, vol.~60, no.~11, pp. 6979--6987, 2014.

\bibitem{yang2019hierarchical}
S.~Yang, A.~Hareedy, R.~Calderbank, and L.~Dolecek, ``Hierarchical coding to
  enable scalability and flexibility in heterogeneous cloud storage,'' 2019,
  arXiv preprint arXiv:1905.02279.

\bibitem{ZLLS2016}
B.~{Zhu}, X.~{Li}, H.~{Li}, and K.~W. {Shum}, ``Replicated convolutional codes:
  A design framework for repair-efficient distributed storage codes,'' in
  \emph{2016 54th Annual Allerton Conference on Communication, Control, and
  Computing, Monticello, IL}, 2016, pp. 1018--1024.

\end{thebibliography}

\end{document}